\documentclass[a4paper,envcountsame]{llncs}
\usepackage{pdfsync}
\usepackage{refcount}
\usepackage{cite}

   \usepackage[pdftex]{graphicx}

\usepackage[cmex10]{amsmath}
\usepackage{amssymb,alltt}
\usepackage[all]{xy}
\xyoption{poly}

\usepackage{stackrel}
\usepackage{xspace}
 \usepackage[usenames]{color}
\usepackage{graphics} 
\usepackage{stmaryrd} 

\usepackage{marginnote}

\newcounter{quotecount}

\renewcommand{\subset}{\subseteq}

\newcommand{\paragrafik}[1]{\vspace{0.1cm} \noindent {\bf #1.}}

\spnewtheorem*{con}{Conjecture}{\bfseries}{\rmfamily}

\spnewtheorem{claim}{Claim}{\bfseries}{\rmfamily}

\newcommand{\tri}[4]{
  \vcenter{\xy \xygraph{!{/r9px/:} 
       [] [rrr]{\scriptscriptstyle{\circlearrowright}} 
          !P3{~={#4}~><{@{.}}~*{\ifcase\xypolynode \or #2 \or #1 \or #3\fi}}}
       \endxy}
}

\newcommand{\fixtri}[4]{
  \vcenter{\xy \xygraph{!{/r9px/:} 
       []  !P3{~={#4}~><{@{.}}~*{\ifcase\xypolynode \or #2 \or #1 \or #3\fi}}}
  \endxy}
}

\newcommand{\wClo}{\mathcal {C}}

\newcommand{\Aa}{\textbf{A}}
\newcommand{\B}{\textbf{B}}
\newcommand{\C}{\mathcal{C}}
\newcommand{\V}{\mathcal{V}}

\newcommand{\K}{\mathcal{K}}

\newcommand{\Q}{\mathbb{Q}}

\newcommand{\Gg}{\mathbb{G}}
\newcommand{\g}{\bold{g}}
\newcommand{\h}{\bold{h}}
\newcommand{\I}{\mathcal{I}}
\newcommand{\Oo}{\mathcal{O}}
\newcommand{\w}{\rightarrow}
\newcommand{\G}{\Gamma}
\newcommand{\rr}{\varrho}
\newcommand{\oo}{\omega}
\DeclareMathOperator{\VCSP}{VCSP}
\DeclareMathOperator{\Opt}{Opt}
\DeclareMathOperator{\Pol}{Pol}
\DeclareMathOperator{\wPol}{wPol}
\DeclareMathOperator{\fPol}{fPol}

\DeclareMathOperator{\IdPol}{IdPol}
\DeclareMathOperator{\wRelClo}{wRelClo}
\DeclareMathOperator{\Feas}{Feas}

\DeclareMathOperator{\Imp}{Imp}

\DeclareMathOperator{\supp}{supp}

\DeclareMathOperator{\id}{id}

\DeclareMathOperator{\Exp}{Exp}
\DeclareMathOperator{\Inv}{Inv}
\DeclareMathOperator{\Const}{Const}
\DeclareMathOperator{\Mjrty}{Mjrty}
\DeclareMathOperator{\Mnrty}{Mnrty}
\DeclareMathOperator{\Hh}{H}
\DeclareMathOperator{\Ss}{S}
\DeclareMathOperator{\Pp}{P}
\DeclareMathOperator{\HSP}{HSP}
\DeclareMathOperator{\HS}{HS}

\DeclareMathOperator{\Mn}{Mn}
\DeclareMathOperator{\Mj}{Mj}

\newif\ifisfinalversion
\isfinalversionfalse

 \ifisfinalversion
 	\newcommand{\notacol}[2]{}
 	\newcommand{\notacoll}[2]{}
  \else
 
	\newcommand{\notacol}[2]{\marginnote{\scriptsize\raggedright\hspace{0pt}\color{#2}\textit{#1}}}
	
  \fi

\begin{document}

\title{Algebraic Properties of Valued Constraint Satisfaction Problem}

\author{Marcin Kozik\inst{1} \and Joanna Ochremiak\inst{2}}

\institute{Jagiellonian University
\and University of Warsaw}

\maketitle

\begin{abstract}
The paper presents an algebraic framework for optimization problems expressible as Valued Constraint Satisfaction Problems. Our results generalize the algebraic framework for the decision version~(CSPs) provided by Bulatov et al.~[SICOMP 2005].

We introduce the notions of weighted algebras and varieties and 
use the Galois connection due to Cohen et al.~[SICOMP~2013] to link VCSP languages to weighted algebras. 
We show that the difficulty of VCSP depends only on the weighted variety generated by the associated weighted algebra.

Paralleling the results for CSPs we exhibit a reduction to cores and rigid cores which allows us to focus on idempotent weighted varieties.
Further, we propose an analogue of the Algebraic CSP Dichotomy Conjecture;
prove the hardness direction and verify that it agrees with known results for
VCSPs on two-element sets~[Cohen et al. 2006], finite-valued VCSPs~[Thapper and \v Zivn\'y 2013] and conservative VCSPs~[Kolmogorov and \v Zivn\' y 2013].

\end{abstract}

\section{Introduction}

An instance of the Constraint Satisfaction Problem (CSP) consists of variables~(to be evaluated in a domain) and constraints restricting the evaluations. 
The aim is to find an evaluation satisfying all the constraints or satisfying the maximal possible number of constraints or approximating the maximal possible number of satisfied constraints etc. depending on the version of the problem.
Further one can divide constraint satisfaction problems with respect to the size of the domain, the allowed constraints or the shape of the instances. 

A particularly interesting version of CSP was proposed in a seminal paper of Feder and Vardi~\cite{FV}. 
In this version a CSP is defined by {\em a language} which consists of a finite number of relations over a finite set. 
An instance of such a CSP is allowed if all the constraint relations are from this set. 
The goal is to determine whether an instance has a solution satisfying all the constraints.

Each language clearly defines a problem in NP; the whole family of problems is interesting for another reason: it is  robust enough to include some well studied computational problems, e.g. 2-colorability, 3-SAT, solving systems of linear equations over ${\mathbb Z}_p$, and still is conjectured~\cite{FV} not to contain problems of intermediate complexity. 
This conjecture is known as the Constraint Satisfaction Dichotomy Conjecture of Feder and Vardi.
Confirming this conjecture would establish CSPs as one of the largest natural subclasses of NP without problems of intermediate complexity.

The conjecture always attracted a lot of attention, but the first results, 
even very interesting ones, were usually very specialized~(e.g.~\cite{HN}).
A major breakthrough appeared with a series of papers establishing {\em the algebraic approach to CSP}~\cite{Jeav,Bulat,BKJ}. 
This deep connection with an independently developed branch of mathematics introduced a new viewpoint and provided tools necessary to tackle wide classes of CSP languages at once.
At the heart of this approach lies a Galois connection between languages and clones of operations called \emph{polymorphisms}~(which completely determine the complexity of the language).

Results obtained using this new methods include a full complexity classifications for CSPs on three-element sets~\cite{Bul3} and containing all unary relations~\cite{Bulatov03,Bulatov}. 
Moreover, the algebraic approach to CSP allowed to propose a boundary between the tractable and NP-complete problems: this conjecture is known as the Algebraic Dichotomy Conjecture.
Unfortunately, despite many efforts~(e.g.~\cite{Bul3}), both conjectures remain open.

The Valued Constraint Satisfaction Problem (VCSP) further extends the approach proposed by Feder and Vardi. 
The role of constraints is played by {\em cost functions} describing the price of choosing particular values for variables as a part of the solution. 
This generalization allows to construct languages modeling standard optimization problems, for example MAX-CUT. 
Moreover, by allowing $\infty$ as a cost of a tuple, a VCSP language can additionally model every problem that CSP can model, as well as hybrid problems like MIN-VERTEX-COVER. This makes the extended framework even more general~(compare the survey~\cite{jkz14:survey}).

A number of classes of VCSPs have been thoroughly investigated. 
The underlying structure suggested capturing the properties of languages of cost functions using an amalgamation of algebraic and numerical techniques~\cite{TZ,CCJK}.
The first approach which provides a Galois correspondence~(mirroring the Galois correspondence for CSPs) was proposed by Cohen et al.~\cite{CJZ}. 
A weighted clone defined in this paper fully captures the complexity of a VCSP language. 

The present paper builds on that correspondence imitating the line of research for CSPs~\cite{BKJ}. It is organized in the following way:
Section 2 contains preliminaries and basic definitions. 
In Section 3 we present a reduction to cores and rigid cores. 
Section 4 introduces a concept of a weighted algebra and a weighted variety, and shows that those notions are well behaved in the context of the Galois connection for VCSP.
Reductions developed in Section 3 together with definitions from Section 4 allow us to focus on idempotent varieties.
Section 5 states a conjecture postulating (for idempotent varieties) the division between the tractable and NP-hard cases of VCSP.
The conjecture is clearly a strengthening  of the Algebraic Dichotomy Conjecture~\cite{BKJ}. 
Section 5 contains additionally the proof of the hardness direction of the conjecture as well as the reasoning showing that the conjecture agrees with complexity classifications for VCSPs on two-element sets~\cite{CCJK}, with finite-valued cost functions~\cite{TZ}, and with conservative cost functions~\cite{KZ}.

\section{Preliminaries}



\subsection{The Valued Constraint Satisfaction Problem}

Throughout the paper, let $\overline{\Q} = \Q \cup \{\infty\}$ denote the set of rational numbers with (positive) infinity. We assume that $x + \infty = \infty$ and $y \cdot \infty = \infty$ for $y \geq 0$. An $r$-ary \emph{relation} on a set $D$ is a subset of $D^r$, a \emph{cost function} on $D$ of arity $r$ is a function from $D^r$ to $\overline{\Q}$. We denote by $\Phi_{D}$ the set of all cost functions on $D$. A cost function which takes only finite values is called \emph{finite-valued}. A $\{0, \infty\}$-valued cost function is called \emph{crisp} and can be viewed as a relation.


\begin{definition}\label{inst}
An \emph{instance of the valued constraint satisfaction problem (VCSP)} is a triple $\I = ( V, D,\C )$ with $V$ a finite set of \emph{variables}, $D$ a finite \emph{domain} and $\C$ a finite multi-set of \emph{constraints}. Each constraint is a pair $C = (\sigma, \varrho)$ with $\sigma$ a tuple of variables of length $r$ and $\varrho$ a cost function on $D$ of arity $r$.

An \emph{assignment} for $\I$ is a mapping $s \colon V \w D$. 
The \emph{cost} of an assignment $s$ is given by $Cost_{\I}(s) = \sum _{(\sigma,\varrho) \in \C} \varrho(s(\sigma))$ (where $s$ is applied component-wise). To solve $\I$ is to find an assignment with a minimal cost, called an \emph{optimal} assignment.
\end{definition}


\begin{example}\label{Max-Cut}
In the \textsc{Max-Cut} problem, one needs to find a partition of the vertices of a given graph into two sets, such that the number of edges with ends in different sets is maximal. This problem is NP-hard.

The \textsc{Max-Cut} problem can be expressed as an instance of VCSP. The domain has two elements $0$ and $1$. Variables in the instance are vertices of the graph and for each edge $e$ there is a constraint of a form $(e, \rr_{XOR})$, where $\rr_{XOR}$ is a binary cost function defined by $$
\rr_{XOR}(x,y) = \begin{cases} 1 &\mbox{if } x=y, \\
0 & \mbox{otherwise.} \end{cases}$$ 
Any assignment of the values $0$ and $1$ to the variables corresponds to a partition of the graph. The cost of an assignment is equal to the number of edges of the graph minus the number of cut edges.
\end{example}

Any set $\G \subseteq \Phi_{D}$ is called a \emph{valued constraint language} over $D$, or simply a \emph{language}. If all cost functions from $\G$ are $\{0, \infty\}$-valued or finite-valued, we call it a \emph{crisp} or \emph{finite-valued} language, respectively. If $\G$ is a language, but not necessarily finite-valued or crisp, we sometimes stress this fact by saying that $\G$ is a \emph{general-valued} language.

By $\VCSP(\G)$ we denote the class of all VSCP instances in which all cost functions in all constraints belong to $\G$. 
$\VCSP(\G_{crisp})$, where $\G_{crisp}$ is the language consisting of all crisp cost functions on some fixed set $D$, is equivalent to the classical CSP.
For an instance $\I \in \VCSP(\G)$ we denote by $\Opt_{\G}(\I)$ the cost of an optimal assignment. 
We say that a language $\G$ is \emph{tractable} if, for every finite subset $\G' \subseteq \G$, there exists an algorithm solving any instance $\I \in \VCSP(\G')$ in polynomial time, and we say that $\G$ is \emph{NP-hard} if $\VCSP(\G')$ is NP-hard for some finite $\G' \subseteq \G$. Example~\ref{Max-Cut} shows that the language $\{ \rr_{XOR} \}$ is NP-hard.


\paragrafik{Weighted Relational Clones} 
We follow the exposition from~\cite{CJZ} and define a closure operator on valued constraint languages that preserves tractability.


\begin{definition}\label{Ex}
A cost function $\rr$ is \emph{expressible} over a valued constraint language $\G \subseteq \Phi_{D}$ if there exists an instance $\I_{\rr} \in \VCSP(\G)$ and a list $(v_1, \dots, v_r)$ of variables of $\I_{\rr}$, such that
$$\rr(x_{1}, \dots, x_{r}) =  \min_{\{  s \colon V \w D \ | \ s(v_i)=x_i \}} Cost_{\I_{\rr}}(s).$$
\end{definition}

Note that the list of variables $(v_1, \dots, v_r)$ in the definition above might contain repeated entries. Hence, it is possible that there are no assignments $s$ such that $s(v_i)=x_i$ for all $i$. We define the minimum over the empty set to be $\infty$.


\begin{definition}\label{WRC}
A  set $\G \subseteq \Phi_{D}$ is a \emph{weighted relational clone} if it is closed under, expressibility, scaling by non-negative rational constants, and addition of rational constants. We define $\wRelClo(\G)$ to be the smallest weighted relational clone containing $\G$.
\end{definition}

If $\rr(x_1, \dots, x_{r}) = \rr_1(y_1, \dots, y_s) + \rr_2(z_1, \dots, z_t)$ for some fixed choice of arguments $y_1, \dots, y_s, z_1, \dots, z_t$ from amongst $x_1, \dots, x_{r}$ then the cost function $\rr$ is said to be obtained by \emph{addition} from the cost functions $\rr_1$ and $\rr_2$. It is easy to see that a weighted relational clone is closed under addition, and minimisation over arbitrary arguments.

The following result shows that we can restrict our attention to languages which are weighted relational clones.

\begin{theorem}[Cohen et al.~\cite{CJZ}]\label{Wrelclo}
A valued constraint language $\G$ is tractable if and only if $\wRelClo(\G)$ is tractable, and it is NP-hard if and only if $\wRelClo(\G)$ is NP-hard.
\end{theorem}


\paragrafik{Weighted polymorphisms}
A $k$-ary \emph{operation} on $D$ is a function $f \colon D^k \w D$. We denote by $\Oo_D$ the set of all finitary operations on $D$ and by $\Oo_D^{(k)}$ the set of all $k$-ary operations on $D$.
The $k$-ary \emph{projections}, defined for all $i \in \{1, \dots, k\}$, are the operations $\pi^{(k)}_i$ such that $\pi^{(k)}_i(x_1, \dots, x_k) = x_i$.
Let $f \in \Oo_D^{(k)}$ and $g_1,\dots, g_k \in \Oo_D^{(l)}$. The $l$-ary operation $f[g_1,\dots, g_k]$ defined by $f[g_1,\dots, g_k](x_1,\dots, x_l)=f(g_1(x_1,\dots, x_l),\dots, g_k(x_1,\dots, x_l))$ is called the \emph{superposition} of $f$ and $g_1,\dots, g_k$.


A  set $C \subseteq \Oo_{D}$ is a \emph{clone of operations} (or simply a \emph{clone}) if it contains all projections on $D$ and is closed under superposition. The set of $k$-ary operations in a clone $C$ is denoted $C^{(k)}$. The smallest possible clone of operations over a fixed set $D$ is the set of all projections on $D$, which we denote $\Pi_{D}$. 

Following~\cite{CJZ} we define a $k$-ary \emph{weighting} of a clone $C$ to be a function $\oo \colon C^{(k)} \w \Q$ such that $\sum_{f \in C^{(k)}} \ \oo(f) = 0$, and if $\oo(f) < 0$ then $f$ is a projection. 
The set of operations to which a weighting $\omega$ assigns positive weights is called the \emph{support} of $\omega$ and denoted $\supp(\omega)$.

A new weighting of the same clone can be obtained by scaling a weighting by a non-negative rational, adding two weightings of the same arity and by the following operation called \emph{superposition}.


\begin{definition}\label{sup}
Let $\oo$ be a $k$-ary weighting of a clone $C$ and let $g_1, \dots, g_k \in C^{(l)}$. A \emph{superposition} of $\oo$ and $g_1, \dots, g_k$ is a function $\oo[g_1, \dots, g_k] \colon C^{(l)} \w \Q$ defined by $$\oo[g_1, \dots, g_k](f') \ = \sum_{\{ f \in C^{(k)} \ | \ f[g_1, \dots, g_k]=f'    \}} \oo(f).$$
\end{definition}

The sum of weights that any superposition $\oo[g_1, \dots, g_k]$ assigns to the operations in $C^{(l)}$ is equal to the sum of weights in $\oo$, which is $0$. However, it may happen that a superposition assigns a negative value to an operation that is not a projection. A superposition is said to be \emph{proper} if the result is a valid weighting.


A non-empty set of weightings over a fixed clone $C$ is called a \emph{weighted clone} if it is closed under non-negative scaling, addition of weightings of equal arity and proper superposition with operations from $C$.
For any clone of operations $C$, the set of all weightings over $C$ and the set of all zero-valued weightings of $C$ are weighted clones. 

We say that an $r$-ary relation $R$ on $D$ is \emph{compatible} with an operation $f \colon D^k \rightarrow D$ if, for any list of $r$-tuples ${\bf x_1, \dots, x_k} \in R$ we have $f({\bf{x_1, \dots, x_k}}) \in R$ (where $f$ is applied coordinate-wise). Let $\rr \colon D^r \rightarrow \overline{\Q}$ be a cost function. We define $\Feas(\rr) = \{ {\bf x} \in D^r \ | \ \rr(\bf{x}) \mbox{ is finite} \}$ to be the \emph{feasibility relation} of~$\rr$. We call an operation $f \colon D^k \rightarrow D$ a \emph{polymorphism} of $\rr$ if the relation $\Feas(\rr)$ is compatible with it.
For a valued constraint language $\G$ we denote by $\Pol(\G)$ the set of operations which are polymorphisms of all cost functions $\rr \in \G$. It is easy to verify that $\Pol(\G)$ is a clone. The set of $m$-ary operations in $\Pol(\G)$ is denoted $\Pol_{m}(\G)$.

For crisp cost functions (relations) this notion of polymorphism corresponds precisely to the standard notion of polymorphism which has played a crucial role in the complexity analysis for the CSP~\cite{Jeav, Bulat}. 


\begin{definition}
Take $\rr$ to be a cost function of arity $r$ on $D$, and let $C \subseteq \Pol(\{ \rr \})$ be a clone of operations. A weighting $\oo \colon C^{(k)} \w \Q$ is called a \emph{weighted polymorphism} of $\rr$ if, for any list of $r$-tuples $\bf x_1, \dots, x_k \in \Feas(\rr)$, we have $$\sum_{f \in C^{(k)}} \oo(f) \cdot \rr(f(\mathbf{x_1, \dots, x_k})) \leq 0.$$
\end{definition} 

For a valued constraint language $\G$ we denote by $\wPol(\G)$ the set of those weightings of the clone $\Pol(\G)$ that are weighted polymorphisms of all cost functions $\rr \in \G$. 
The set of weightings $\wPol(\G)$ is a weighted clone~\cite{CJZ}. 

\begin{example}
For any lattice-ordered set $D$,
a function $\rr \colon D^{r} \w \Q$ is called \emph{submodular}  if for all $\mathbf{x_1, x_2} \in D^{r}$ $$\rr(\min(\mathbf{x_1,x_2})) + \rr(\max(\mathbf{x_1,x_2})) - \rr(\mathbf{x_1}) - \rr(\mathbf{x_2}) \leq 0.$$ The above condition can be equivalently expressed by saying that the set of submodular functions on $D$ is the set of cost functions with a binary weighted polymorphism $\oo$, defined as follows: $$
\oo(f) = \begin{cases} -1 &\mbox{if } f \mbox{ is a projection,} \\
 \ \ 1 & \mbox{if } f \mbox{ is one of the operations } \min \mbox{ or } \max, \\
\ \ 0 & \mbox{otherwise.} \end{cases}$$
\end{example}

An operation $f$ is \emph{idempotent} if $f(x,...,x) = x$. A weighted polymorphism is called \emph{idempotent} if all operations in its support are idempotent.

An operation $f \in \Oo_D^{(k)}$ is \emph{cyclic} if for every $x_1, \dots , x_k \in D$ we have that $f(x_1,x_2, \dots ,x_k) = f(x_2, \dots, x_k ,x_1)$. 
A weighted polymorphism is called \emph{cyclic} if its support is non-empty and contains cyclic operations only.

A cost function $\rr$ is said to be \emph{improved} by a weighting $\oo$ if $\oo$ is a weighted polymorphism of $\rr$. For any set $W$ of weightings over a fixed clone $C \subseteq \Oo_{D}$ we denote by $\Imp(W)$ the set of cost functions on $D$ which are improved by all weightings $\oo \in W$.
The following result, together with Theorem~\ref{Wrelclo}, implies that tractable valued constraint languages can be characterised by weighted polymorphisms.

\begin{theorem}[Cohen et al.~\cite{CJZ}]
For any finite valued constraint language $\G$, we have $\Imp(\wPol(\G)) = \wRelClo(\G)$.
\end{theorem}

For more information on the valued constraint satisfaction problem see the recent survey~\cite{jkz14:survey}.

\subsection{Algebras and varieties}


In this subsection we introduce the basic concepts of universal algebra that serve us as tools later on in this paper. 
An \emph{algebraic signature} is a set of function symbols together
with (finite) arities. An \emph{algebra} $\Aa$ over a fixed signature $\Sigma$ consists of a set $A$, called the \emph{universe} of $\Aa$, and a set of \emph{basic operations} that correspond to the symbols in the signature, i.e., if the signature contains a $k$-ary symbol $f$ then the algebra has a basic operation $f^{\Aa}$, which is a function $f^{\Aa} \colon A^{k} \w A$.

A subset $B$ of the universe of an algebra $\Aa$ is a \emph{subuniverse} of $\Aa$ if it is closed under all operations of $\Aa$. An algebra $\B$  is a \emph{subalgebra} of $\Aa$ if $B$ is a subuniverse of $\Aa$ and the operations of $\B$ are restrictions of all the operations of $\Aa$ to $B$.
Let $(\Aa_{i})_{i \in I}$ be a family of algebras (over the same signature). Their \emph{product} $\Pi_{i \in I} \Aa_{i}$ is an algebra with the universe equal to the cartesian product of the $A_{i}$'s  and operations computed coordinate-wise. 
For two algebras $\Aa$ and $\B$ (over the same signature), a
\emph{homomorphism} from $\Aa$ to $\B$ is a function $h \colon A \w B$ that preserves all operations. It is easy to see, that an image of a homomorphism $h \colon A \w B$ is a subalgebra of $\B$.


Let $\K$ be a class of algebras over a fixed signature $\Sigma$. We denote by $\Ss(\K)$ the class of all subalgebras of algebras in $\K$, by $\Pp(\K)$ the class of all products of algebras in $\K$, by $\Pp_{fin}(\K)$ the class of all finite products, and by $\Hh(\K)$ the class of all homomorphic images of algebras in $\K$. If $\K=\{\Aa\}$ we write $\Ss(\Aa)$, $\Pp(\Aa)$, and $\Hh(\Aa)$ instead of $\Ss(\{\Aa\})$, $\Pp(\{\Aa\})$, and $\Hh(\{\Aa\})$, respectively.

Similarly $\V(\K)$ is the smallest class of algebras closed under all three operations.
For an algebra $\Aa$ the variety $\V(\{\Aa\})$~(denoted $\V(\Aa)$) is the variety \emph{generated} by $\Aa$, and $\V_{fin}(\Aa)$ is the class of finite algebras in $\V(\Aa)$.
The variety $\V(\Aa)$ can be characterised as follows:

\begin{proposition}[Tarski~\cite{T}]\label{hsp}
For any finite algebra $\Aa$, we have $$\V(\Aa) = \HSP(\Aa) \text{ \ \ and \ \ } \V_{fin}(\Aa) = \HSP_{fin}(\Aa).$$
\end{proposition} 

We say that an equivalence relation $\sim$ on $A$ is a \emph{congruence} of $\Aa$
if the following condition is satisfied for all operations $f$ of $\Aa$: if for all  $i \in \{1, \dots, k\}$, we have $a_i \sim b_i$, then $f(a_1, \dots, a_k) \sim  f(b_1, \dots, b_k),$ where $k$ is the arity of $f$. 
Every congruence $\sim$ of $\Aa$ determines a \emph{quotient} algebra $\Aa /{\sim}$. Its universe is the set of the equivalence classes $A /{\sim}$ and operations are defined using their arbitrarily chosen representatives.


A \emph{term} $t$ in a signature $\Sigma$ is a formal expression built from variables and symbols in $\Sigma$ that syntactically describes the composition of basic operations. For an algebra $\Aa$ over $\Sigma$ a \emph{term operation} $t^{\Aa}$ is an operation obtained by composing the basic operations of $\Aa$ according to $t$. Let $s$ and $t$ be a pair of terms in a signature $\Sigma$. We say that $\Aa$ satisfies the \emph{identity} $s \approx t$ if the term operations $s^{\Aa}$ and $t^{\Aa}$ are equal.
We say that a class of algebras $\V$ over $\Sigma$ satisfies the identity $s \approx t$ if every algebra in $\V$ does.



\begin{theorem}[Birkhoff~\cite{B}]\label{Bir}
A class of algebras $\V$ is a variety if and only if there exists a set of identities such that $\V$ contains precisely those algebras that satisfy all the identities from this set.
\end{theorem}

It follows from Birkhoff's theorem that the variety $\V(\Aa)$ is the class of algebras that satisfy all the identities satisfied by $\Aa$. Moreover, if $\Aa$ is finite then $\V(\Aa)$ is \emph{locally finite}, i.e., every finitely generated algebra in $\V(\Aa)$ is finite.

\section{Core Valued Constraint Languages}\label{sec:core}


For each valued constraint language $\G$ there is an associated algebra. It has universe $D$ and the set of operations $\Pol(\G)$. If all polymorphisms of $\G$ are idempotent it means that the algebra $(D, \Pol(\G))$ satisfies the identity $f(x,\ldots, x) \approx x$ for every operation $f$. Such algebras are called \emph{idempotent}. 
In this section we prove that every finite valued constraint language has a computationally equivalent valued constraint language whose associated algebra is idempotent. 

\subsection{Positive Clone.}

Those polymorphisms of a given language $\G$ which are assigned a positive weight by some weighted polymorphisms $\oo \in \wPol(\G)$ are of special interest in the rest of the paper. We begin this section by proving that they form a clone.

Let $\wClo$ be a weighted clone over a set $D$.
The following proposition shows that the set $\bigcup_{\oo \in \wClo} \supp(\oo)$, together with the set of projections $\Pi_{D}$, is a clone. We call it the \emph{positive clone} of $\wClo$ and denote by $C^{+}$~(if $\wClo$ is $\wPol(\G)$ then $C^+$ is denoted by $\Pol^+(\G)$).


\begin{proposition}\label{prop:posclo}
If $\wClo$ is a weighted clone then $C^+$ is a clone.
\end{proposition}
We will use the following technical lemma (Lemma 6.5 from~\cite{CJZ}). It implies that any weighting that can be expressed as a weighted sum of arbitrary superpositions can also be expressed as a superposition of a weighted sum of proper superpositions.

\begin{lemma}\label{superp}
Let $\wClo$ be a weighted clone, and let $\oo_1$ and $\oo_2$ be weightings in $\wClo$, of arity $k$ and $l$ respectively. For any $m$-ary operations $f_1, \dots ,f_k, g_1, \dots ,g_l$ of $C$:
$$c_1 \oo_1[f_1, \dots ,f_k] + c_2  \oo_2[g_1, \dots ,g_l] =
\oo[f_1, \dots ,f_k, g_1, \dots ,g_l],$$
where $$\oo = c_1 \omega_1[\pi_1^{(k+l)}, \ldots, \pi_k^{(k+l)}] + c_2 \omega_2[\pi_{k+1}^{(k+l)}, \ldots, \pi_{k+l}^{(k+l)}].$$
\end{lemma}
\begin{proof}
We need to show that the set $C^{+}$ is closed under superposition. Take a $k$-ary operation $f$ and a list of $l$-ary operations $g_1, \dots, g_k$ that all belong to $C^+$.  

If $f$ is a projection there is nothing to prove. 
Otherwise there is a weighting $\omega\in\wClo$ such that $\omega(f)>0$. 
Similarly for each $g_i$ which is not a projection we find $\omega_i$ such that $\omega_i(g_i)>0$~(if $g_i$ is a projection we put $\omega_i$ to be the zero-valued $l$-ary weighting).

Now, there exist non-negative rational numbers $i_j$ such that the sum
$$
\omega[g_1,\dotsc,g_k] + i_1\omega_1[\pi_1^l,\dotsc,\pi_l^l]+\dotsb + i_k\omega_k[\pi_1^l,\dotsc,\pi_l^l]
$$
is a valid weighting. 
By Lemma~\ref{superp} this weighting can be obtained as a superposition of a sum of proper superpositions and therefore belongs to $\wClo$ which finishes the proof.

\end{proof}



\subsection{Cores.}

Let $\G$ be a valued constraint language with a domain $D$. For $S \subseteq D$ we denote by $\G[S]$ the valued constraint language defined on a domain $S$ and containing the restriction of every cost function $\rr \in \G$ to $S$. We show that $\G$ has a computationally equivalent valued constraint language $\G'$ such that $\Pol_1^+(\G')$ contains only bijective operations. Such a language is called a \emph{core}. Moreover, $\G'$ can be chosen to be equal $\G[S]$ for some $S \subseteq D$.

%

\begin{proposition}\label{core}
For every valued constraint language $\G$ there exists a core language $\G'$, such that the valued constraint language $\G$ is tractable if and only if $\G'$ is tractable, and it is NP-hard if and only if $ \G'$ is NP-hard.
\end{proposition}

We prove the above result by generalizing the arguments for finite-valued languages given in~\cite{HKP,TZ}. We need an auxiliary lemma.

\begin{lemma}\label{opt}
For a valued constraint language $\G$, let $f \in \Pol_1^{+}(\Gamma)$ and let $\I \in \VCSP(\G)$. If $s$ is an optimal assignment for $\I$, then $f(s)$ is also optimal.
\end{lemma}

\begin{proof}
Let $f$, $\I$ and $s$
be like in the statement of the lemma. Observe that $Cost_{\I}$ (see Definition~\ref{inst}) can be seen as a cost function whose arity is equal to the number of variables in $\I$. Moreover, $Cost_{\I}$ belongs to $\wRelClo(\G)$ as it is clearly expressible over $\G$. 
If $Cost_{\I}(s) = \infty$ then there is no assignment with a finite cost and we are done. 

Assume that $Cost_{\I}(s) < \infty$, which means that $s \in \Feas(Cost_{\I})$. 
If $f \neq \id$, then there exists a weighted polymorphism $\oo$ with $\oo(f) > 0$. By definition the following inequality is satisfied: $$\sum_{g \in \Pol_{1}(\G)} \oo(g) \cdot Cost_{\I}(g(s)) \leq 0.$$  
Without loss of generality we can assume that $\oo(\id) = -1$. Then we have that $\sum_{g \in \supp(\oo)} \oo(g) = 1$ and the inequality above can be rewritten as
$$\sum_{g \in \supp(\oo)} \oo(g) \cdot Cost_{\I}(g(s))  \leq Cost_{\I}(s).$$
On the other hand,
$$\sum_{g \in \supp(\oo)} \oo(g) \cdot Cost_{\I}(g(s)) \geq  \sum_{g \in \supp(\oo)} \oo(g) \cdot Cost_{\I}(s) = Cost_{\I}(s).$$
Therefore $Cost_{\I}(g(s)) = Cost_{\I}(s)$ for each operation $g \in \supp(\oo)$. Since $f \in \supp(\oo)$ and $s$ is optimal, $f(s)$ is also optimal.
\end{proof}

\begin{proof}
(of Proposition~\ref{core}) Let $\G$ be a valued constraint language over a domain $D$. Suppose that there is a unary polymorphism $f \in \Pol^{+}(\G)$ that is not bijective. Let $\G' = \G[f(D)]$, where $f(D)\varsubsetneq D$ denotes the range of $f$. There is a natural correspondence between instances of $\VCSP(\G')$ and instances of $\VCSP(\G)$, induced by the correspondence between functions in $\G$ and their restrictions in $\G'$. For any instance $\I'$ of $\VCSP(\G')$ the corresponding instance $\I$ of $\VCSP(\G)$ has the same variables. The cost function $\rr'$ in each constraint is replaced by any cost function $\rr$ from $\G$, which is equal to $\rr'$ when restricted to $f(D)$. We show that $\Opt_{\G}(\I) = \Opt_{\G'}(\I')$.

Any assignment for $\I'$ is also an assignment for $\I$, and hence $\Opt_{\G}(\I) \leq \Opt_{\G'}(\I')$. Furthermore, by Lemma~\ref{opt} for each $s$ that is an optimal assignment for $\I$, we have $$Cost_{\I}(s) = Cost_{\I}(f(s)) = Cost_{\I'}(f(s)).$$ Therefore, $\Opt_{\G}(\I) \geq \Opt_{\G'}(\I')$.

It follows that $\VCSP(\G)$ is tractable if and only if $\VCSP(\G')$ is tractable, and it is NP-hard if and only if $\VCSP(\G')$ is NP-hard. Moreover, the valued constraint language $\G'$ is defined over a smaller domain.
We replace $\G$ with $\G'$ and repeat this procedure, until we obtain a language $\G'$ that is a core.
\end{proof}

For core languages we characterize the set of unary weighted polymorphisms.

\begin{proposition}\label{unary}
Let $\G$ be a core valued constraint language. A unary weighting $\oo$ is a weighted polymorphism of $\G$ if and only if it assigns positive weights only to such bijective operations $f \in \Pol_{1}(\G)$ that, for all cost functions $\rr \in \G$, satisfy $\rr \circ f = \rr$.
\end{proposition}

\begin{proof}
If a valid unary weighting $\oo$ assigns positive weights only to such operations $f \in \Pol_{1}(\G)$ that, for all cost functions $\rr \in \G$, satisfy $\rr \circ f = \rr$, then for each $\rr \in \G$ and a tuple $\mathbf{x} \in \Feas(\rr)$
$$\sum_{f \in \Pol_{1}(\G)} \oo(f) \cdot \rr(f(\mathbf{x})) = \sum_{f \in \Pol_{1}(\G)} \oo(f) \cdot \rr(\mathbf{x})  = 0,$$
and $\oo$ is clearly a weighted polymorphism of $\G$.

For the other direction, let $\oo$ be a unary weighted polymorphism of $\G$, such that $\supp(\oo) \neq \emptyset$. Without loss of generality assume that $\oo(\id)=-1$. Since $\G$ is a core language, the operations $g \in \supp(\oo)$ are bijective. For $\rr \in \G$ and a tuple $\mathbf{x} \in \Feas(\rr)$ for which $\rr$ takes the minimal value, we have
$$\sum_{g \in \supp(\oo)} \oo(g) \cdot \rr(g(\mathbf{x})) + \oo(\id) \cdot \rr(\mathbf{x}) \leq 0, \text{ hence}$$
$$\rr(\mathbf{x}) \geq \sum_{g \in \supp(\oo)} \oo(g) \cdot \rr(g(\mathbf{x})) \geq \sum_{g \in \supp(\oo)} \oo(g) \cdot \rr(\mathbf{x}) = \rr(\mathbf{x}).$$
Therefore $\rr(g(\mathbf{x})) = \rr(\mathbf{x})$ for each $g \in \supp(\oo)$, which means that the operations in the support preserve the minimal weight.

Note that, since each $g \in \supp(\oo)$ is bijective, it determines a bijection of the set $\Feas(\rr)$. We have shown that this bijection preserves the set of tuples with minimal weight. It can be similarly shown by induction that it preserves the set of tuples with any other fixed weight. Hence, we have proved that $\rr \circ g = \rr$ for all $g \in \supp(\oo)$.
\end{proof}

This implies that, for any core language $\G$, a unary polymorphism belongs to $\Pol^+(\G)$ if and only if it is bijective and preserves all cost functions in~$\G$.

Let $\G$ be a finite core valued constraint language over a domain $D$. For each arity $m$ we fix an enumeration of all the elements of $D^m$. This allows us to treat every $m$-ary operation $f \in \Oo_D^{(m)}$ as a $|D^m|$-tuple.
We define a $|D^m|$-ary cost function in $\wRelClo(\G)$ that precisely distinguishes the $m$-ary operations in the positive clone from all the other $m$-ary polymorphisms. To do this we need the following technical lemma, which is a variant of the well known Farkas' Lemma used in linear programming:


\begin{lemma}[Farkas~\cite{farkas}]\label{Farkas}
Let $S$ and $T$ be finite sets of indices, where $T$ is a disjoint union of two subsets, $T_{\geq}$ and $T_{=}$. For all $i \in S$, and all $j \in T$, let $a_{i,j}$ and $b_{j}$ be rational numbers. Exactly one of the following holds:
\begin{itemize}
\item Either there exists a set of non-negative rational numbers $\{z_i \ | \ i \in S \}$ and a rational number $C$ such that
$$\text{for each } j \in R_{\geq}, \ \ \sum_{i \in S} a_{i,j} z_{i} \geq b_{j} + C,$$
$$\text{for each } j \in R_{=}, \ \ \sum_{i \in S} a_{i,j} z_{i} = b_{j} + C.$$
\item Or else there exists a set of rational numbers $\{y_{j} \ | \ j \in T \}$ such that $\sum_{j \in T} y_j = 0$ and
$$\text{for each } j \in T_\geq, \ \ y_j \geq 0,$$
$$\text{for each } i \in S, \ \ \sum_{j \in T} y_j a_{i,j} \leq 0,$$ 
$$\text{and } \sum_{j \in T} y_j b_j > 0.$$
\end{itemize}
\end{lemma}

The set $\{y_j \ | \ j \in T \}$ defined in the lemma is called a \emph{certificate of unsolvability}.


\begin{proposition}\label{polymor}
Let $\G$ be a finite core valued constraint language over a domain $D$. For every $m$ there exists a cost function $\rr \colon \Oo_D^{(m)} \rightarrow \overline{Q}$ in $\wRelClo(\G)$, and a rational number $P$, such that for every $f \in \Oo_D^{(m)}$ the following conditions are satisfied:
\begin{enumerate}
\item $\rr (f) \geq P$,
\item $\rr (f) < \infty$ if and only if $f \in \Pol(\G)$,
\item $\rr (f) = P $ if and only if $f \in \Pol^{+}(\G)$.
\end{enumerate}
\end{proposition}

\begin{proof}
The cost function $\rr$ is given by a sum of all cost functions in $\G$ with positive coefficients that we define later on. 

Like in the classical CSP, a cost function whose feasibility relation contains exactly those $|D^m|$-tuples which are $m$-ary polymorphisms of $\G$ is defined by:
  \begin{equation*}
    \sum_{\substack{\rr\in\Gamma\\ ({\bf a_1},\dotsc, {\bf a_m})\in (\Feas(\rr))^m}} \rr(x_{{\bf b_1}},\dotsc, x_{{\bf b_{r}}}),
  \end{equation*}
where ${\bf b_i}(j) = {\bf a_j}(i)$, and $r$ is the arity of $\rr$.
For each summand we introduce a variable $z_{\rr,{\bf a_1},\dotsc, {\bf a_m}}$ and, for each $f\in \Pol^+_m(\Gamma)$ we write:
  \begin{equation*}
    \sum_{\substack{\rr\in\Gamma\\ ({\bf a_1},\dotsc, {\bf a_m}) \in (\Feas(\rr))^m}}z_{\rr,{\bf a_1},\dotsc, {\bf a_m}}  \rr(f({\bf b_1}),\dotsc,f({\bf b_{r}}) ) = 0 + C,
  \end{equation*}
  while for each $f\in \Pol_m(\Gamma)\setminus \Pol^+_m(\Gamma)$:
  \begin{equation*}
    \sum_{\substack{\rr\in\Gamma\\ ({\bf a_1},\dotsc, {\bf a_m}) \in (\Feas(\rr))^m}}z_{\rr,{\bf a_1},\dotsc, {\bf a_m}}  \rr(f({\bf b_1}),\dotsc,f({\bf b_{r}}) ) \geq 1 + C,
  \end{equation*}
where ${\bf b_i}(j) = {\bf a_j}(i)$, and $r$ is the arity of $\rr$.

By putting the above equalities and inequalities together we obtain a system of linear inequalities and equations. By Lemma~\ref{Farkas} there are two mutually exclusive possibilities. 
First, there may exist a set of non-negative rational numbers $z_{\rr,{\bf a_1},\dotsc, {\bf a_m}}$ and a rational  number $C$, such that this system is satisfied. Then the proposition is proved: items 1. and 3. follow trivially from construction. Item 2. follows by definition of the cost function.

Otherwise, there exists a set $\{ y_f \ | \ f \in \Pol_m(\G) \}$ which forms the certificate of unsolvability. Then let us consider a weighting defined by $\oo(f) = y_f$. If $\oo$ is a valid weighting, then it is an $m$-ary weighted polymorphism of $\G$. Moreover, $\oo$ assigns to all operations in $\Pol_m(\G) \setminus \Pol^{+}_m(\G)$ non-negative weights that sum up to a positive number. Hence, for some $h \in \Pol_m(\G) \setminus \Pol^{+}_m(\G)$, we have $\oo(h) > 0$, which contradicts $h \notin \Pol^{+}_m(\G)$.
If it happens that $y_g < 0$ for some operation $g \in \Pol^{+}_m(\G)$ that is not a projection, then there exists an $m$-ary weighted polymorphism of $\G$ which assigns a positive weight to $g$. By scaling it and adding to $\oo$~(as in the proof of Proposition~\ref{prop:posclo}), we obtain the weighted polymorphism needed for the contradiction. 
\end{proof}

\subsection{Rigid cores.}\label{rigid}

We further reduce the class of languages that we need to consider. Let $\G$ be a core valued constraint language over an $n$-element domain $D = \{d_1, \dots, d_n \}$. For each $i \in \{ 1, \dots, n \}$, let 
$$
N_i(x) = \begin{cases} 0 &\mbox{if } x=d_i, \\
\infty & \mbox{otherwise.} \end{cases}$$ 
and let $\G_c$ denote the valued constraint language obtained from $\G$ by adding all cost functions $N_i$. Observe that $\Pol(\G_c) = \IdPol(\G)$, where by $\IdPol(\G)$ we denote the set of idempotent polymorphisms of the language $\G$.
Hence, the only unary polymorphism of $\G_c$ is the identity, which also means that there is only one unary weighted polymorphism of $\G_c$ -- the zero-valued polymorphism.


\begin{definition}
A valued constraint language $\G$ is a \emph{rigid core} if there is exactly one unary polymorphism of $\G$, which is the identity. 
\end{definition}

The notion of rigid core corresponds to the classical notion of rigid core considered in CSP~\cite{BKJ}. A valued constraint language $\G$ is a rigid core if the set of feasibility relations of all cost functions from $\G$ is a rigid core in the standard sense, which is also equivalent to all polymorphisms of $\G$ being idempotent.

We now prove a result which, together with Proposition~\ref{core}, implies that for each finite language $\G$, there is a computationally equivalent language that is a rigid core.

\begin{proposition}\label{rigid core}
Let $\G$ be a valued constraint language which is finite and a core.
The valued constraint language $\G_c$ is a rigid core. 
Moreover, $\G$ is tractable if and only if $\G_c$ is tractable, 
and $\G$ is NP-hard if and only if $\G_c$ is NP-hard.
\end{proposition}

\begin{proof}
Let $\G$ be a finite core valued constraint language over a domain $D = \{d_1, \dots, d_n \}$. 
It follows from Proposition~\ref{polymor} that there exist an $n$-ary cost function $N \in \wRelClo(\G)$, and positive rational numbers $P < Q$, such that the following conditions are satisfied:
\begin{itemize}
\item $N(x_1, \dots, x_n) = P $ if and only if the unary operation $g$ defined by $d_i \mapsto  x_i$ belongs to $\Pol^{+}(\G)$,
\item $N(x_1, \dots, x_n) > Q $ if and only if the unary operation $g$ defined by $d_i \mapsto  x_i$ belongs to $\Pol(\G) \setminus \Pol^{+}(\G)$,
\item otherwise $N(x_1, \dots, x_n) = \infty$.
\end{itemize}
Assume without loss of generality that $N \in \G$. We show a polynomial-time Turing reduction from $\VCSP(\G_c)$ to $\VCSP(\G)$.

Let $\I_c = (V_c,D, \C_c)$ be an instance of $\VCSP(\G_c)$. The set of variables $V$ in the new instance $\I$ is a disjoint union of $V_c$ and 
$\{ v_1, \dots, v_n \}$.
For every constraint of the form $((v),N_i)$ in $\C_c$ we:
\begin{itemize}
\item add a constraint $((v,v_i),\rr_=)$, where 
$$
\rr_=(x,y) = \begin{cases} 0 &\mbox{if } x=y, \\
\infty & \mbox{otherwise} \end{cases}$$
(this cost function is expressible over every valued constraint language, so without loss of generality we can assume that $\rr_= \in \G$),
\item remove the constraint $((v),N_i)$ from $\C_c$.
\end{itemize}
We obtain a new set of constraints $\C_1$, where all cost functions are already from~$\G$.


Let $C$ be the sum of weights that all cost functions in all constraints in $\C_1$ assign to all tuples in their feasibility relations. The final set of constraints $\C$ additionally contains $m$ constraints of the form $((v_1, \dots, v_n), N)$, where $m$ is big enough to ensure that $m \cdot (Q-P) > C$.

There are three possibilities:
\begin{itemize}
\item If $\Opt_{\G}(\I) = \infty$ then no assignment for $\I_c$ has a finite cost. Suppose otherwise and let $s_c$ be an assignment for $\I_c$ with a finite cost. Then $s_c$ gives rise to an assignment $s$ for $\I$ with a finite cost. It coincides with $s_c$ on $V_c$ and for each $i \in \{1, \dots, n\}$, we set $s(v_i) = d_i$.
\item The optimal assignment $s$ for $\I$ satisfies $N(s(v_1, \dots, v_n))=P$. Then the tuple $s(v_1, \dots, v_n)$ determines a unary operation $g$, defined by $d_i \mapsto s(v_i)$. The operation $g$, by the definition of the cost function $N$, belongs to the positive clone $\Pol^{+}(\G)$. Hence, $g^{-1}$ also belongs to the positive clone. Since $\G$ is a core, the assignment $g^{-1}(s)$ is optimal for $\I$. Its restriction onto $V_c$ is an optimal assignment for $\I_c$.
\item The optimal assignment $s$ for $\I$ satisfies $N(s(v_1, \dots, v_n)) > Q$. While there are $m$ constraints of the form $((v_1, \dots, v_n), N)$, we have $$Cost_{\I}(s) \geq m \cdot Q > m \cdot P + C.$$ If there was any assignment $s_c$ for $\I_c$ with a finite cost, the corresponding assignment $s$ for $\I$ would satisfy $Cost_{\I}(s) < m \cdot P +C$, which gives a contradiction, and implies that $\Opt_{\G_c}(\I_c) = \infty$.
\end{itemize}
\end{proof}

If $\G$ is a core language then the positive clone of $\G_c$ contains precisely the idempotent operations from the positive clone of $\G$. To show this, we first prove the following lemma:

\begin{lemma}\label{idemp}
Let $\G$ be a core valued constraint language. For every weighted polymorphism $\oo \in \wPol(\G)$ there exists an idempotent weighted polymorphism $\oo' \in \wPol(\G)$ such that $\supp(\oo) \cap \IdPol(\G) \subseteq \supp(\oo')$. Moreover, if $\oo$ is cyclic then $\oo'$ can be chosen to be cyclic.
\end{lemma}

\begin{proof}
Consider a weighted polymorphism $\oo \in \wPol(\G)$. Take a non-idempotent operation $g \in supp(\oo)$ and let $h$ be a unary operation defined by $h(x) = g(x, \ldots, x)$. Since $\Pol^+(\G)$ is a clone of operations, $h \in \Pol^+(\G)$. Then by Proposition~\ref{unary} the operation $h$ is bijective and preserves all cost functions in $\G$. We modify the weighted polymorphism $\oo$ by adding $\oo(g)$ to the weight of the idempotent operation $h^{-1} \circ g$ and then assigning weight $0$ to the operation $g$. It is straightforward to check that the new weighting is a weighted polymorphism of $\G$. If $g$ is cyclic then so is $h^{-1} \circ g$. Hence if $\oo$ is cyclic then so is the new weighting. We repeat this construction for every non-idempotent operation in $\supp(\oo)$. Finally, we obtain an idempotent weighted polymorphism $\oo'$ which satisfies the conditions of the lemma.
\end{proof}

\begin{proposition}\label{idpolpol}
Let $\G$ be a valued constraint language which is a core. Then $\IdPol^+(\G) = \Pol^+(\G_c)$.
\end{proposition}

\begin{proof}
Clearly both sets contain all the projections.
Let us take $f \in \Pol^+(\G_c)$ that is not a projection and let $\oo$ be a weighted polymorphism of $\G_c$ such that $f \in \supp(\oo)$. There is a corresponding weighted polymorphism $\oo'$ of $\G$, which is equal to $\oo$ on the idempotent operations and equal $0$ otherwise. Then we have $f \in \supp(\oo')$. Since $f$ is idempotent it follows that $f \in \IdPol^+(\G)$.

To prove the reverse inclusion consider $f \in \IdPol^+(\G)$ that is not a projection. Let $\oo$ be a weighted polymorphism of $\G$ such that $f \in \supp(\oo)$. By Lemma~\ref{idemp} there exists an idempotent weighted polymorphism $\oo'$ of $\G$ such that $\supp(\oo) \cap \IdPol(\G) \subseteq \supp(\oo')$. The weighting $\oo''$, defined as a restriction of $\oo'$ to the idempotent operations, is a weighted polymorphism of $\G_c$ with $f \in \supp(\oo'')$. Hence $f \in \Pol^+(\G_c)$.
\end{proof}

\section{Weighted varieties}

One of the fundamental results of the algebraic approach to CSP~\cite{BKJ,Bulat,Larose20} says that the complexity of a crisp language $\G$ depends only on the variety generated by the algebra $(D, \Pol(\G))$. We generalize this fact to VCSP.



A $k$-ary \emph{weighting} $\omega$ of an algebra $\Aa$ is a function that assigns rational weights to all $k$-ary term operations of $\Aa$ in such a way, that the sum of all weights is $0$, and if $\omega(f) < 0$ then $f$ is a projection. A (\emph{proper}) \emph{superposition} $\oo[g_1, \dots, g_k]$ of a weighting $\oo$ with a list of $l$-ary term operations $g_1, \dots, g_k$ from $\Aa$ is defined the same way as for clones (see Definition~\ref{sup}). 
An algebra $\Aa$ together with a set of weightings closed under non-negative scaling, addition of weightings of equal arity and proper superposition with operations from $\Aa$ is called a \emph{weighted algebra}.

For a variety $\V$ over a signature $\Sigma$ and a term $t$ we denote by $[t]_{\V}$ the equivalence class of $t$ under the relation $\approx_{\V}$ such that $t \approx_{\V} s$ if and only if the variety $\V$ satisfies the identity $t \approx s$ (we skip the subscript, writing $[t]$ instead $[t]_{\V}$, whenever the variety is clear from the context). Observe that if the variety is locally finite then there are finitely many equivalence classes of terms of a fixed arity~\cite{BS}.

\begin{definition}\label{wvariety}
Let $\V$ be a locally finite variety over a signature $\Sigma$. A $k$-ary \emph{weighting} $\oo$ of $\V$ is a function that assigns rational weights to all equivalence classes of $k$-ary terms over $\Sigma$ in such a way, that the sum of all weights is $0$, and if $\oo([t]) < 0$ then $\V$ satisfies the identity $t(x_1, \dots , x_k) \approx x_i$ for some $i \in \{1, \dots, k\}$.
The variety $\V$ together with a nonempty set of weightings is called a \emph{weighted variety}.
\end{definition}

Take any finite algebra $\B \in \V$. A $k$-ary weighting $\oo$ of $\V$ \emph{induces} a weighting $\oo^\B$ of $\B$ in a natural way: $$\oo^{\B}(f) = \sum_{\{[t] \ | \ t^\B = f \}} \oo([t]).$$ If $\oo([t]) < 0$ then the term operation $t^\B$ is a projection, and hence the weighting $\oo^{\B}$ is proper. For a weighted variety $\V$, by $\B \in \V$ we mean the algebra $\B$ together with the set of weightings induced by $\V$.

For every weighting $\oo$ of a finite weighted algebra $\Aa$ there is a corresponding weighting $\oo$ of the variety $\V(\Aa)$ defined by $\oo([t]) = \oo(t^{\Aa})$. It follows from Birkhoff's theorem (see Theorem~\ref{Bir}) that it is well defined. A weighted variety $\V(\Aa)$ \emph{generated} by a weighted algebra $\Aa$ is the variety $\V(\Aa)$ together with the set of weightings corresponding to the weightings of $\Aa$. The correspondence is one-to-one so for simplicity we often identify the weightings of $\V(\Aa)$ with the weightings of $\Aa$.

We prove that every finite algebra $\B \in\V(\Aa)$ together with the set of weightings induced by $\V(\Aa)$ is a weighted algebra. It is straightforward to check its closure under non-negative scaling and addition of weightings of equal arity. We only need to show that $\B$ is closed under proper superpositions.

\begin{proposition}\label{wc}
For a finite weighted algebra $\Aa$ over a fixed signature $\Sigma$ and a finite algebra $\B \in \V(\Aa)$ let $\oo^{\B}$ be a $k$-ary weighting of $\B$ induced by the weighted variety $\V(\Aa)$. If for some list $f_1^{\B}, \ldots, f_k^{\B}$ of $l$-ary term operations from $\B$ the weighting $\oo^{\B}[f_1^{\B}, \ldots, f_k^{\B}]$ is proper then it is induced by some weighting of $\V(\Aa)$.
\end{proposition}

In the proof we use Gordan's Theorem (which is a straightforward consequence of Lemma~\ref{Farkas}).

\begin{theorem}[Gordan~\cite{gordan}]\label{Gordan}
Let $S$ and $T$ be finite sets of indices. For all $i \in S$, and all $j \in T$, let $a_{i,j}$ be rational numbers. Exactly one of the following holds:
\begin{itemize}
\item Either there exists a set of non-negative rational numbers $\{z_i \ | \ i \in S \}$ such that
$$\text{for some } i \in S, \ \ z_i >0,$$
$$\text{for each } j \in T, \ \ \sum_{i \in S} a_{i,j} z_{i} =0.$$
\item Or else there exists a set of rational numbers $\{y_{j} \ | \ j \in T \}$ such that
$$\text{for each } i \in S, \ \ \sum_{j \in T} y_j a_{i,j} > 0.$$ 
\end{itemize}
\end{theorem}

\begin{proof} (of Proposition~\ref{wc})
Let $\oo^{\B}$ and $f_1^{\B}, \ldots, f_k^{\B}$ be as in the statement of the proposition. Assume that the weighting $\oo^{\B}[f_1^{\B}, \ldots, f_k^{\B}]$ is proper. 

Notice that the following conditions are equivalent:
\begin{itemize}
\item the operation $f_i^\B$ is the projection $\pi_j$ on the $j$-th coordinate,
\item there exists a term $t$ such that $t^{\B} = f_i^{\B}$ and $t^\Aa$ is the projection $\pi_j$ on the $j$-th coordinate.
\end{itemize}  
For each $i \in \{1, \dots, k\}$ consider the set $F_i$ of equivalence classes of terms over $\Sigma$ defined by
$$
F_i = \begin{cases} \{[t] \ | \ t^\Aa = \pi_j \} &\mbox{if } f_i^\B = \pi_j, \\
\{[t] \ | \ t^{\B} = f_i^{\B}\} & \mbox{otherwise} \end{cases}$$
(observe that if $f_i^\B$ is a projection then $F_i$ contains a single equivalence class).
Take $\oo$ to be some $k$-ary weighting of $\Aa$ that induces $\oo^{\B}$, and let $$W = \{ \oo[t_1^\Aa, \ldots, t_k^\Aa] \ | \ [t_i] \in F_i \}.$$

Suppose that for some choice of equivalence classes $[t_i] \in F_i$ the superposition $\oo[t_1^\Aa, \ldots, t_k^\Aa]$ is proper.
The weighting $\oo[t_1^\Aa, \ldots, t_k^\Aa]$ of $\Aa$ induces a weighting of $\B$ which is equal to $\oo^{\B}[f_1^{\B}, \ldots, f_k^{\B}]$, thus in this case the proof is concluded.

This shows that a superposition of $\oo^{\B}$ with any list of projections is always induced by some weighting of $\Aa$.

Now let us deal with the case when none of the weightings in $W$ is proper. Without loss of generality we can assume that the operations $f_1^{\B}, \ldots, f_k^{\B}$ are pairwise distinct (otherwise we replace $\oo^{\B}$ by its superposition with a suitable list of projections) and hence the sets $F_i$ are disjoint. Let $F= \bigcup F_i$. We remove from $F$ the element of $F_i$ if $f_i^\B$ is a projection. 
The removed elements cannot cause a problem and therefore we assume that for every $[t] \in F$ the operation $t^\Aa$ is not a projection. We apply Gordan's Theorem to the following system of linear equations:
$$\sum_{\nu \in W}  \nu(t^\Aa) \cdot z_\nu - z_{[t]} =0, \mbox{ for each } [t] \in F.$$

If this system has a non-zero solution in non-negative rational numbers then $z_\nu > 0$ for some $\nu \in W$. Observe that the weighting $\upsilon = \sum_{\nu \in W}  \nu \cdot z_\nu$ is proper. Indeed, by the definition of a superposition the only non-projections that could be assigned negative weights by $\upsilon$ are the operations $t^{\Aa}$ where $[t] \in F$. But each such operation $t^{\Aa}$ is assigned a non-negative weight $z_{[t]}$. Hence, by Lemma~\ref{superp} the weighting $\upsilon$ is equal to a proper superposition of some weighting of $\Aa$ with a list of $l$-ary term operations of $\Aa$. Finally, let $p = \sum_{\nu \in W} z_\nu > 0$. The weighting ${1\over p}  \upsilon$ of $\Aa$ induces a weighting of $\B$ which is equal to $\oo^{\B}[f_1^{\B}, \ldots, f_k^{\B}]$.

Otherwise, there exists a set $\{ y_{[t]} \ | \ [t] \in F \}$ of rational numbers, such that $$\mbox{for each } \nu \in W, \ \ \sum_{[t] \in F} y_{[t]} \cdot \nu(t^\Aa) > 0,$$ and $y_{[t]} <0$ for each $[t] \in F$. For every $i \in \{1, \dots, k\}$ let us choose $[t_i] \in F_i$ satisfying $y_{[t_i]} = \max\{ y_{[t]} \  | \ [t] \in F_i \}$ (if $f_i^\B$ is a projection then we choose $[t_i] \in F_i$ to be the only element of $F_i$ and put $y_{[t_i]}=0$) and consider the weighting $\upsilon = \oo[t_1^\Aa, \ldots, t_k^\Aa]$. Notice that $\upsilon$ may assign negative weights only to operations $t_i^\Aa$. Since $$\sum_{[t] \in F_1} y_{[t]} \cdot \upsilon(t^\Aa) + \dots + \sum_{[t] \in F_k} y_{[t]} \cdot \upsilon(t^\Aa) >0,$$ then $\sum_{[t] \in F_i} y_{[t]} \cdot \upsilon(t^\Aa) > 0$ for some $i \in \{1, \dots, k\}$. Hence $$0 < \sum_{[t] \in F_i} y_{[t]} \cdot \upsilon(t^\Aa) \leq \sum_{[t] \in F_i} y_{[t_i]} \cdot \upsilon(t^\Aa) =y_{[t_i]} \cdot \sum_{[t] \in F_i} \upsilon(t^\Aa).$$ It follows that $\sum_{[t] \in F_i} \upsilon(t^\Aa) < 0$, which is a contradiction, since $\sum_{[t] \in F_i} \upsilon(t^\Aa)$ is the weight that the proper weighting $\oo^{\B}[f_1^{\B}, \ldots, f_k^{\B}]$ assigns to the operation $f_i^\B$ (which is not a projection).
\end{proof}

For a finite weighted algebra $\Aa$ let $\Imp(\Aa)$ denote the set of those cost functions on $A$ that are improved by all weightings of $\Aa$. We prove that for each finite weighted algebra $\B \in \V(\Aa)$ the valued constraint language $\Imp(\B)$ is not harder then $\Imp(\Aa)$. The proof consists of a sequence of lemmas.


\begin{lemma}
Let $\Aa$ be a finite weighted algebra. For any $\B \in  P_{fin}(\Aa)$, there is a polynomial-time reduction of $\VCSP(\Imp(\B))$ to $\VCSP(\Imp(\Aa))$.
\end{lemma}

\begin{proof}
Let $A^n$ be the universe of $\B$ and let $\G$ be a finite subset of $\Imp(\B)$. Take $\rr \in \G$ to be an $r$-ary cost function. There is a natural way of defining a corresponding cost function of arity $n \cdot r$ on the set $A$. We denote this cost function by $\rr'$. 

Let $\oo$ be a $k$-ary weighting of the weighted algebra $\Aa$. The corresponding $k$-ary weighting $\oo^{\B}$ of $\B$ is a weighted polymorphism of $\rr$. Then it is not hard to show that $\oo$ is a weighted polymorphism of $\rr'$, as the basic operations of $\B$ are the operations of $\Aa$ computed coordinate-wise. Hence, each weighting of $\Aa$ is a weighted polymorphism of $\rr'$, which means that $\rr' \in \Imp(\Aa)$.

For each $\rr \in \G$ we have defined a corresponding $\rr' \in \Imp(\Aa)$. Let $\G' \subseteq \Imp(\Aa)$ be the (finite) set of all those cost functions. 

Now take an arbitrary instance $\I = (V, A^n, \C)$ of $\VCSP(\G)$. Replace the domain $A^n$ by $A$, and each variable $v_i \in V$ by a set of $n$ variables $\{v_i^1, \dots, v_i^n  \}$, obtaining a new set of variables $V'$. In each constraint $(\sigma, \rr) \in \C$, where $\rr$ is an $r$-ary cost function, replace the $r$-tuple $\sigma$ of variables from $V$ by the corresponding $n r$-tuple of variables from $V'$, and the cost function $\rr$ by the corresponding cost function $\rr'$ from $\G'$. The new instance $\I'=(V',A,\C')$ is an instance of $\VCSP(\G')$. It is easy to see that there is a one-to-one correspondence between the optimal assignments for $\I$ and the optimal assignments for $\I'$.
\end{proof}


\begin{lemma}
Let $\Aa$ be a finite weighted algebra. For any $\B \in  S(\Aa)$, there is a polynomial-time reduction of $\VCSP(\Imp(\B))$ to $\VCSP(\Imp(\Aa))$.
\end{lemma}

Notice that $\Imp(\B) \subseteq \Imp(\Aa)$, so there is nothing to be proved.


\begin{lemma}
Let $\Aa$ be a finite weighted algebra. For any $\B \in  H(\Aa)$, there is a polynomial-time reduction of $\VCSP(\Imp(\B))$ to $\VCSP(\Imp(\Aa))$.
\end{lemma}

\begin{proof}
By the isomorphism theorem we can consider $\B$ to be a quotient algebra $\Aa / {\sim}$ rather than a homomorphic image of $\Aa$.
Let $A /{\sim}$ be the universe of $\B$ and let $\G$ be a finite subset of $\Imp(\B)$. Take $\rr \in \G$ to be a $r$-ary cost function. We define a corresponding cost function $\rr'$ of arity $r$ on the set $A$ by $\rr'(x_1, \dots, x_r) = \rr([x_1]_{\sim}, \dots [x_r]_{\sim})$. 

Let $\oo$ be a $k$-ary weighting of the weighted algebra $\Aa$. The corresponding $k$-ary weighting $\oo^{\B}$ of $\B$ is a weighted polymorphism of $\rr$. It is not hard to show that $\oo$ is a weighted polymorphism of $\rr'$. Hence, each weighting of $\Aa$ is a weighted polymorphism of $\rr'$, which means that $\rr' \in \Imp(\Aa)$.

For each $\rr \in \G$ we have defined a corresponding $\rr' \in \Imp(\Aa)$. Let $\G' \subseteq \Imp(\Aa)$ be the (finite) set of all those cost functions. 

Now take an arbitrary instance $\I = (V, A /{\sim}, \C)$ of $\VCSP(\G)$. Replace the domain $A /{\sim}$ by $A$. In each constraint $(\sigma, \rr) \in \C$ replace the cost function $\rr$ by a corresponding cost function $\rr'$ from $\G'$. The new instance $\I'=(V,A,\C')$ is an instance of $\VCSP(\G')$. 

If $s' \colon V \w A$ is an optimal assignment for $\I'$, then $s \colon V \w A /{\sim}$ defined by $s(v) = [s'(v)]_{\sim}$ is an optimal assignment for $\I$. On the other hand, if $s \colon V \w A /{\sim}$ is an optimal assignment for $\I$, then any assignment $s' \colon V \w A$, such that for each $v \in V$, we have $s'(v) \in s(v)$, is optimal for $\I'$.
\end{proof}

The above lemmas together with Proposition~\ref{hsp}, imply the following:


\begin{proposition}\label{iii}
For any finite weighted algebra $\Aa$, and any finite $\B \in \V(\Aa)$, there is a polynomial-time reduction of $\VCSP(\Imp(\B))$ to $\VCSP(\Imp(\Aa))$.
\end{proposition}

For a valued constraint language $\G$ the weighted algebra $(D, \wPol(\G))$ is the algebra $(D, \Pol(\G))$ together with the set of weightings $\wPol(\G)$. By Proposition~\ref{iii} the complexity of $\G$ depends only on the weighted variety generated by the weighted algebra $(D, \wPol(\G))$.

\section{Dichotomy conjecture}

An operation $t$ of arity $k$ is called a \emph{Taylor operation} of an algebra (or a variety), if $t$ is
idempotent and for every $j \leq k$ it satisfies an identity of the form $$t(\square_1,\square_2, \dots,\square_k) \approx t(\triangle_1,\triangle_2, \dots,\triangle_k),$$
where all $\square_i$’s and $\triangle_i$’s are substituted with either $x$ or $y$, but $\square_j$ is $x$ whenever $\triangle_j$ is $y$. In this section we work towards a proof of the following theorem:

\begin{theorem}\label{taylor}
Let $\G$ be a finite core valued constraint language. If $\Pol^{+}(\G)$ does not have a Taylor operation, then $\G$ is NP-hard.
\end{theorem}

We conjecture\footnote{The conjecture was suggested in a conversation by Libor Barto, however it might have appeared independently earlier.}
that these are the only cases of finite core languages which give rise to NP-hard VCSPs.

\begin{con}
\emph{Let $\G$ be a finite core valued constraint language. If $\Pol^{+}(\G)$ does not have a Taylor operation, then $\G$ is NP-hard. Otherwise it is tractable.}
\end{con}

For crisp languages $\Pol^{+}(\G) = \Pol(\G)$. Therefore Theorem~\ref{taylor} generalizes the well-known result of Bulatov, Jeavons and Krokhin~\cite{BKJ, Bulat} concerning crisp core languages. Similarly the above conjecture is a generalization of The Algebraic Dichotomy Conjecture for CSP. Later on we show that it is supported by all known partial results on the complexity of VCSPs.

To prove Theorem~\ref{taylor} we use the following characterization of algebras possessing a Taylor operation:

\begin{theorem}[Taylor~\cite{Taylor}]
Let $\Aa$ be a finite idempotent algebra, then the following are equivalent:
\begin{itemize}
\item $\Aa$ has a Taylor operation,
\item $\V(\Aa)$ (equivalently $\HS(\Aa)$) does not contain a two-element algebra whose every term operation is a projection.
\end{itemize}
\end{theorem}

First let us prove an auxiliary lemma.

\begin{lemma}\label{relation}
Let $\G$ be a finite core valued constraint language over a domain $D$, and let $R$ be an $r$-ary relation which is compatible with every polymorphism from $\Pol^+(\G)$. Then there exists a cost function $\rr_R$ in $\wRelClo(\G)$, and a rational number $P$, such that for every $r$-tuple $\bf x$ the following conditions are satisfied:
\begin{itemize}
\item $\rr_R ({\bf x}) \geq P$ and
\item $\rr_R ({\bf x}) = P$ if and only if ${\bf x} \in R$.
\end{itemize}
\end{lemma}

\begin{proof}
  Let $R=\{{\bf x_1},\dotsc,{\bf x_m}\}$ be a relation as in the statement of the lemma.
  By Proposition~\ref{polymor} there exists a cost function $\rr' \colon \Oo_D^{(m)} \rightarrow \overline{Q}$ in $\wRelClo(\G)$, and a rational number $P$, such that for every $f \in \Oo_D^{(m)}$:
\begin{itemize}
\item $\rr' (f) \geq P$,
\item $\rr' (f) < \infty$ if and only if $f \in \Pol(\G)$,
\item $\rr' (f) = P $ if and only if $f \in \Pol^{+}(\G)$.
\end{itemize}
Consider the coordinates ${\bf b_1},\dotsc,{\bf b_r}$, such that ${\bf b_i}(j) = {\bf x_j}(i)$. Minimising the cost function $\rr'$ over all the other coordinates we obtain a cost function $\rr$ satisfying the given conditions.
\end{proof}

\begin{proof}
(of Theorem~\ref{taylor})
Let $\G$ be a finite core valued constraint language over a domain $D$, and let $\G_c$ be a rigid core of $\G$ as defined in Subsection~\ref{rigid}. Suppose that $\Pol^{+}(\G)$ does not have a Taylor operation.  By Proposition~\ref{idpolpol} we have that $\IdPol^+(\G) = \Pol^+(\G_c)$. Therefore, $\Pol^+(\G_c)$ does not have a Taylor operation. Below we prove that $\VCSP(\G_c)$ is NP-hard. 
This implies, by Proposition~\ref{rigid core},
that $\VCSP(\G)$ is NP-hard which concludes the proof.

Let $\Aa$ denote the idempotent algebra $\Pol^+(\G_c)$ over the universe $D$. By Taylor's theorem $\HS(\Aa)$ contains a two-element algebra $\B$ whose every term operation is a projection. By the isomorphism theorem we can consider $\B$ to be a quotient algebra rather than a homomorphic image of a subalgebra of $\Aa$. In other words, there exists a binary relation $S$ compatible with $\Aa$, which is an equivalence relation on some subuniverse $D'$ of $D$ and has two equivalence classes $[d_0]_{S}$ and $[d_1]_{S}$. Moreover, the term operations defined on the set of equivalence classes of $S$ using their arbitrarily chosen representatives are all projections.

Every relation is compatible with a two-element algebra whose every term operation is a projection. Consider the relation $R=\{(1,0,0), (0,1,0), (0,0,1)\}$. It corresponds to the \textsc{One-in-Three Sat} problem, which is NP-complete~\cite{Schaefer}. We define a ternary relation $S_{\textsc{1in3}}$ on $D'$ by: $$S_{\textsc{1in3}} = \{ (x_1, x_2, x_3) \ \colon \text{exactly one of } x_1, x_2, x_3 \text{ belongs to } [d_1]_S \}.$$ This relation is compatible with $\Aa$. Hence there exists a cost function $\rr_{S_{\textsc{1in3}}}$ in $\wRelClo(\G_c)$ satisfying the conditions given by Lemma~\ref{relation}.

Now, for every instance of \textsc{One-in-Three Sat} it is easy to construct~(in polynomial time) an instance of $\VCSP(\G_c)$ such that
if the instance of \textsc{One-in-Three Sat} has a solution then this solution gives rise to the minimal evaluation of the constructed instance.
This finishes the reduction.
\end{proof}


As the Taylor operation is difficult to work with, in the following we use a characterization of Taylor algebras as the algebras possessing a cyclic term.

\begin{theorem}[Barto and Kozik~\cite{Kozik}]\label{thm:BK}
A finite idempotent
algebra $\Aa$ has a Taylor operation if and only if it has a cyclic operation
if and only if it has a cyclic operation of arity $p$, for every prime $p > |\Aa|$.
\end{theorem}

Since $(D,\IdPol^+(\G))$ is a finite idempotent algebra it follows that:

\begin{corollary}
For any finite core valued constraint language $\G$, $\Pol^{+}(\G)$ has a Taylor operation if and only if it has an idempotent cyclic operation.
\end{corollary}

\subsection{Two-element domain}

A complete complexity classification for valued constraint languages over a two-element domain was established in~\cite{CCJK}. All tractable languages have been defined via multimorphisms, which are a more restricted form of weighted polymorphisms. A $k$-ary \emph{multimorphism} of a language $\G$, specified as a $k$-tuple $\langle f_1, \ldots, f_k \rangle$ of $k$-ary operations on $D$, is a $k$-ary weighted polymorphism $\oo$ of $\G$ such that $\oo = {1\over k}\big(\sum_i f_i - \sum_i \pi^{(k)}_i\big)$.


An operation $f \in \Oo_D^{(3)}$ is called a \emph{majority} operation if for every $x,y \in D$ we have that $f (x, x, y) = f (x, y, x) = f (y, x, x) = x$. Similarly, an operation $f \in \Oo_D^{(3)}$ is called a \emph{minority} operation if for every $x,y \in D$ it satisfies $f(x,x,y) = f(x,y,x) = f(y,x,x) = y$. Observe that on a two-element domain there is precisely one majority operation, which we denote by $\Mjrty$, and precisely one minority operation,  which we denote by $\Mnrty$.

\begin{theorem}[Cohen et al.~\cite{CCJK}]
Let $\G$ be a core valued constraint language on $D = \{0, 1\}$. If $\G$ admits at least one of the following six multimorphisms, then $\G$ is tractable. Otherwise it is NP-hard.
\begin{enumerate}
\item $ \langle \min, \min \rangle$,
\item $ \langle \max, \max \rangle$,
\item $ \langle \min,\max \rangle$,
\item $ \langle \Mjrty, \Mjrty, \Mjrty \rangle$, 
\item $ \langle \Mnrty, \Mnrty, \Mnrty \rangle$,
\item $ \langle \Mjrty, \Mjrty, \Mnrty \rangle$.
\end{enumerate}
\end{theorem}

In~\cite{CCJK} the complexity classification is given for languages that are not necessarily cores. It is not difficult to prove, though, that the general case is equivalent to the above theorem. This is because every language over a two-element domain which is not a core is tractable.

We show that the dichotomy conjecture for VCSP agrees with the complexity classification for valued constraint languages over a two-element domain.

\begin{proposition}\label{two}
Let $\G$ be a finite core valued constraint language on $D = \{0, 1\}$. Then $\Pol^{+}(\G)$ has an idempotent cyclic operation if and only if $\G$ admits at least one of the following six multimorphisms.
\begin{enumerate}
\item $ \langle \min, \min \rangle$,
\item $ \langle \max, \max \rangle$,
\item $ \langle \min,\max \rangle$,
\item $ \langle \Mjrty, \Mjrty, \Mjrty \rangle$, 
\item $ \langle \Mnrty, \Mnrty, \Mnrty \rangle$,
\item $ \langle \Mjrty, \Mjrty, \Mnrty \rangle$.
\end{enumerate}
\end{proposition}

On $D = \{0,1\}$ there are precisely two constant operations, which
we denote by $\Const_0$ and $\Const_1$. By $\Inv$ we denote the \emph{inversion} operation defined by $\Inv(0) =1$ and $\Inv(1)=0$. To prove Proposition~\ref{two} we use the following theorem: 

\begin{theorem}[Cohen et al.~\cite{CJZ}]\label{nonzero}
Let $W$ be a weighted clone on $D = \{0, 1\}$ that contains a weighting which assigns positive weight to at least one operation that is not a projection. Then $W$ contains one of the following nine weightings:
\begin{enumerate}
\item $ \langle \Const_0 \rangle$,
\item $ \langle \Const_1 \rangle$,
\item $ \langle \Inv \rangle$,
\item $ \langle \min, \min \rangle$,
\item $ \langle \max, \max \rangle$,
\item $ \langle \min,\max \rangle$,
\item $ \langle \Mjrty, \Mjrty, \Mjrty \rangle$, 
\item $ \langle \Mnrty, \Mnrty, \Mnrty \rangle$,
\item $ \langle \Mjrty, \Mjrty, \Mnrty \rangle$.
\end{enumerate}
\end{theorem}

\begin{proof}
(of Proposition~\ref{two})
Each of the operations $\min$, $\max$, $\Mjrty$ and $\Mnrty$ is idempotent and cyclic. If $\G$ admits at least one of the six multimorphisms listed in the statement of the proposition then obviously $\Pol^{+}(\G)$ has an idempotent cyclic operation.

For the other direction, let $\G$ be a finite core valued constraint language on $D = \{0, 1\}$ such that $\Pol^{+}(\G)$ has an idempotent cyclic operation. Let $\G_c$ be a rigid core of $\G$ as defined in Subsection~\ref{rigid}. By Proposition~\ref{idpolpol} we have that $\IdPol^+(\G) = \Pol^+(\G_c)$. Therefore, the weighted clone $\wPol(\G_c)$ contains a weighting which assigns positive weight to at least one operation that is not a projection. Then $\wPol(\G_c)$ contains one of the nine weightings listed in Theorem~\ref{nonzero}. Since the first three of them are not idempotent, it follows that $\G_c$, and hence $\G$, admits one of the six remaining multimorphisms, which finishes the proof.
\end{proof}

\subsection{Finite-valued languages}

\begin{theorem}[Thapper and \v{Z}ivn\'{y}~\cite{TZ}] 
Let $\G$ be a finite-valued constraint language which is a core. If $\G$ admits an idempotent cyclic weighted polymorphism of some arity $m>1$, then $\G$ is tractable. Otherwise it is NP-hard.
\end{theorem}

To show that our conjecture agrees with the above complexity classification we prove the following result (which holds for general-valued languages):

\begin{proposition}\label{char}
Let $\G$ be a core valued constraint language. Then $\G$ admits an idempotent cyclic weighted polymorphism of some arity $m>1$ if and only if $\Pol^{+}(\G)$ contains an idempotent cyclic operation of the same arity.
\end{proposition}

One implication is strightforward: if $\G$ admits an idempotent cyclic weighted polymorphism of some arity $m>1$ then $\Pol^{+}(\G)$ contains an idempotent cyclic operation. To show the other implication we use a technique of constructing weighted polymorphism introduced in~\cite{KolTZ}. 

The construction of a new weighted polymorphism of arity $m$ is based on grouping operations in $\Oo_D^{(m)}$ into so-called \emph{collections} and working with weightings that assign the same weight to every operation in a collection. 

Let $\Gg$ be a fixed set of collections, i.e., subsets of $\Oo_D^{(m)}$, and let $\Gg^* \subseteq \Gg$ be a set of collections satisfying some desired property. An \emph{expansion operator} $\Exp$ takes a collection $\g \in \Gg$ and produces a probability distribution $\delta$ over $\Gg$. We say that $\Exp$ is \emph{valid} for a language $\G$ if, for any $\rr \in \G$ and any $\g \in \Gg$, the probability distribution $\delta = \Exp(\g)$ satisfies
$$\sum_{\h \in \Gg} \sum_{h \in \h} {\delta(\h) \over |\h|} \rr(h({\bf x_1, \dots, x_m})) \leq \sum_{g \in \g} {1\over |\g|} \rr(g({\bf x_1, \dots, x_m})),$$
for any ${\bf x_1, \dots, x_m} \in \Feas(\rr)$. We say that the operator $\Exp$ is \emph{non-vanishing} (with respect to the pair $(\Gg, \Gg^*)$) if, for any $\g \in \Gg$, there exists a sequence of collections $\g_0, \g_1, \dots , \g_r$ with $\g_0 = \g$, such that for each $i \in \{0, \dots , r-1\}$ the collection $\g_{i+1}$ is assigned a non-zero probability by $\Exp(\g_i)$, and $\g_r \in \Gg^*$.

\begin{lemma}[“Expansion Lemma”~\cite{KolTZ}]\label{expansion}
Let $\Exp$ be an expansion operator which is valid for the language $\G$ and non-vanishing with respect to $(\Gg,\Gg^*)$. If $\G$ admits a weighted polymorphism $\oo$ with $\supp(\oo) \subseteq \bigcup \Gg$, then it also admits a weighted polymorphism $\oo^*$ with $\supp(\oo^*) \subseteq \bigcup \Gg^*$.
\end{lemma}

\begin{proof} 
(of Proposition~\ref{char})
In order to show the remaining implication assume that $\Pol^{+}(\G)$ contains an idempotent cyclic operation $f$ of arity $m>1$. There exists a weighted polymorphism $\oo$ of $\G$ such that $f \in \supp(\oo)$. 

We define $\sim$ to be the smalles equivalence relation on $\Oo_D^{(m)}$ such that $g \sim g'$ if $g(x_1,x_2, \dots ,x_k) = g'(x_2, \dots, x_k ,x_1)$. Observe that if $g \sim g'$ and $g \in \Pol(\G)$ then also $g' \in \Pol(\G)$. Let $\Gg$ consist of the equivalence classes of the relation $\sim$ restricted to $\Pol(\G)$, and let $\Gg^* \subseteq \Gg$ be the set of all one-element equivalence classes, i.e., each $\g \in \Gg^*$ contains a single cyclic operation.

We now define the expansion operator $\Exp$. Take an arbitrary $\g \in \Gg\setminus\Gg^*$~(for $\g\in\Gg^*$ we produce the probability distribution choosing $\g$ with probability $1$) and choose a single operation $g \in \g$. Notice that $\g = \{g_1, g_2, \ldots, g_m\}$, where $g_1 = g$ and $g_i(x_1, \ldots, x_n) = g(x_i, x_{i+1}, \ldots, x_{i-1})$ for $i \in \{2, \ldots, m\}$. Consider a weighting
$$\nu = c \cdot (\oo[g_1, g_2, \dots, g_m] + \oo[g_2, \dots, g_m, g_1] + \dots +  \oo[g_m, g_1, \dots, g_{m-1}]),$$
where $c$ is a suitable positive rational, which we define later on\footnote{Note that the definition of $\nu$ does not depend on the choice of $g$ from $\g$.}.

The weighting $\nu$ assigns a positive weight to a cyclic operation $f[g_1, g_2, \dots, g_m]$. This proves that $\nu$ is not zero-valued, and hence in the above definition $c$ can be chosen so that the sum of positive weights in $\nu$ equals $1$.  
We say that a weighting $\oo$ is \emph{weight-symmetric} if $\oo(g) = \oo(g')$ whenever $g \sim g'$. It is easy to check that $\nu$ is weight-symmertic. We define $\Exp(\g)$ to be a probability distribution $\delta$ on $\Gg$ such that
$$
\delta(\h)  = \begin{cases} |\h| \cdot \nu(h) &\mbox{if } \h \subseteq \supp(\nu), \\
0 & \mbox{otherwise,} \end{cases}$$ 
where $h$ is any of the operations in $\h$.
We have already pointed out that $\nu$ assigns a positive weight to a cyclic operation $f[g_1, g_2, \dots, g_m]$. It follows that $\Exp(\g)$ assigns non-zero probability to the one-element equivalence class $\{f[g_1, g_2, \dots, g_m]\}$. Therefore, $\Exp$ is non-vanishing. 

It remains to show that $\Exp$ is valid for $\G$. Observe that $\nu$ assigns negative weights only to the operations in $\g$.
Since it is weight-symmetric, $\nu(g) = - {1\over |\g|}$ for every $g \in \g$. The weighting $\nu$ might not be valid but it is not difficult to see that it satisfies the condition characterizing weighted polymorphisms, i.e., for any cost function $\rr \in \G$, and any list of tuples $\bf x_1, \dots, x_m \in \Feas(\rr)$, we have
\begin{align*}
\sum_{f \in \Pol_m(\G)} \nu(f) \cdot \rr(f(\mathbf{x_1, \dots, x_m})) & \leq 0, \mbox{ hence } \\
\sum_{h \in \supp(\nu)} \nu(h) \cdot \rr(h(\mathbf{x_1, \dots, x_m})) & \leq \sum_{g \in \g} {1\over |\g|} \cdot \rr(g({\bf x_1, \dots, x_m})), \mbox{ but} \\
\sum_{h \in \supp(\nu)} \nu(h) \cdot \rr(h(\mathbf{x_1, \dots, x_m})) & = \sum_{\h \in \Gg} \sum_{h \in \h} {\delta(\h) \over |\h|} \cdot \rr(h({\bf x_1, \dots, x_m})). 
\end{align*}
This proves that $\Exp$ is valid, so by Lemma~\ref{expansion} the language $\G$ admits a weighted polymorphism $\oo^*$ whose support contains only cyclic $m$-ary operations. Moreover, it follows from the proof of the Expansion Lemma in~\cite{KolTZ} that $\oo^*$ can be constructed so that $f \in \supp(\oo^*)$. By Lemma~\ref{idemp} there exists an idempotent weighted polymorphism $\oo'$ of $\G$ such that $\supp(\oo^*) \cap \IdPol(\G) \subseteq \supp(\oo')$. Its support is non-emply and contains cyclic operations only. This concludes the proof.
\end{proof}

\subsection{Conservative languages}

A valued constraint language $\G$ over a domain $D$ is called \emph{conservative} if it contains all $\{0,1\}$-valued unary cost functions on $D$.
An operation $f \in \Oo_D^{(k)}$ is \emph{conservative} if for every $x_1, \dots , x_k \in D$ we have that $f(x_1,\ldots,x_k)\in \{ x_1,\ldots,x_k\}$, and
a weighted polymorphism is \emph{conservative} if its support contains conservative operations only.

A \emph{Symmetric Tournament Pair} (\emph{STP}) is a conservative binary multimorphism $ \langle \sqcap, \sqcup \rangle$, where both operations are commutative, i.e., $\sqcap(x,y)=\sqcap(y,x)$ and $\sqcap(x,y)=\sqcap(y,x)$ for all $x,y \in D$, and moreover $\sqcap(x,y) \neq \sqcup(x,y)$ for all $x \neq y$.
A \emph{MJN} is a ternary conservative multimorphism $\langle \Mj_1,\Mj_2,\Mn_3\rangle$, such that $\Mj_1,\Mj_2$ are majority operations, and $\Mn_3$ is a minority operation.

\begin{theorem}[Kolmogorov and \v{Z}ivn\'{y}~\cite{KZ}] Let $\G$ be a conservative constraint language over a domain $D$. If $\G$ admits a conservative binary multimorphism $\langle \sqcap, \sqcup \rangle$ and a conservative ternary multimorphism $\langle \Mj_1,\Mj_2,\Mn_3\rangle$, and there is a family $M$ of two-element subsets of $D$, such that:
\begin{itemize}
\item for every $\{x,y\} \in M$, $\langle \sqcap, \sqcup \rangle$ restricted to $\{x,y\}$ is an STP,
\item for every $\{x,y\} \not \in M$, $\langle \Mj_1,\Mj_2,\Mn_3\rangle$ restricted to $\{x,y\}$ is an MJN,
\end{itemize}
then $\G$ is tractable. Otherwise it is NP-hard.
\end{theorem}

Observe that every weighted polymorphism of a conservative language $\G$ is conservative. Indeed, consider a $k$-ary weighted polymorphism $\oo \in \wPol(\G)$ and take any $x_1, \ldots, x_k \in D$. Let $\rr \in \G$ be a unary cost function such that $\rr(x_i)=0$ for $i \in \{1, \ldots, k\}$ and $\rr(x)=1$ otherwise. Then 
$$\sum_{g \in \supp(\oo)} \oo(g) \cdot \rr(g(x_1, \ldots, x_k)) = \sum_{g \in \Pol_1(\G)} \oo(g) \cdot \rr(g(x_1, \ldots, x_k)) \leq 0,$$
hence for each $g \in \supp(\oo)$ we have that $\rr(g(x_1, \ldots, x_k))=0$, so $g(x_1, \ldots, x_k) \in \{ x_1,\ldots,x_k\}$. This implies that the positive clone of a conservative language is idempotent, and hence every conservative language is a core. 

We show that our conjecture agrees with the above complexity classification:

\begin{proposition}\label{charcon}
Let $\G$ be a conservative constraint language over a domain $D$. Then $\Pol^{+}(\G)$ has an idempotent cyclic operation if and only if $\G$ admits a conservative binary multimorphism $\langle \sqcap, \sqcup \rangle$ and a conservative ternary multimorphism $\langle \Mj_1,\Mj_2,\Mn_3\rangle$, and there is a family $M$ of two-element subsets of $D$, such that:
\begin{itemize}
\item for every $\{x,y\} \in M$, $\langle \sqcap, \sqcup \rangle$ restricted to $\{x,y\}$ is an STP,
\item for every $\{x,y\} \not \in M$, $\langle \Mj_1,\Mj_2,\Mn_3\rangle$ restricted to $\{x,y\}$ is an MJN.
\end{itemize}
\end{proposition}

\begin{proof}
Let $\G'$ be the language $\G$ together with all $\{0,\infty\}$-valued unary cost functions on $D$. For every weighted polymorphism $\oo \in \wPol(\G)$ there is a corresponding weighted polymorphism of $\G'$, which is equal to $\oo$ on the conservative operations. Therefore $\Pol^+(\G')=\Pol^+(\G)$, and $\G$ admits a conservative multimorphism $\langle f_1, \ldots, f_k \rangle$ if and only if $\G'$ does.
Now let $\rr : D \rightarrow \overline{\Q}$ be any general-valued unary cost function. Observe that $\rr \in \wRelClo(\G')$. It follows that without loss of generality we can assume that $\G$ contains all general-valued unary cost functions. We do so in the rest of the proof. Observe that every polymorphism of such language is conservative.

Assume that $\G$ admits the two conservative multimorphisms $\langle \sqcap, \sqcup \rangle$ and $\langle \Mj_1,\Mj_2,\Mn_3\rangle$ described in the statement of the proposition. There are only four idempotent operations on a two-element domain $\{x,y\}$, namely: $\max$, $\min$, $\pi_1$ and $\pi_2$ (we assume that $\{x,y\}=\{0,1\}$). Each of the operations $\sqcap$, $\sqcup$ restricted to any two-element subset $\{x,y\}$ of $D$ must be equal to one of those four. Therefore it is not difficult to prove, using $\{0,1\}$-valued unary cost functions, that for every two-element subset $\{x,y\}$ of $D$:
\begin{itemize}
\item either $\langle \sqcap, \sqcup \rangle$ restricted to $\{x,y\}$ is an STP,
\item or $\sqcap$ restricted to $\{x,y\}$ is equal to $\pi_1$ and $\sqcup$ restricted to $\{x,y\}$ is equal to $\pi_2$ (possibly the other way round).
\end{itemize}
Let $M'$ be the set of those two-element subsets $\{x,y\}$ of $D$, for which $\langle \sqcap, \sqcup \rangle$ restricted to $\{x,y\}$ is an STP. Obviously $M \subseteq M'$.

Let $t(x,y,z) = ((x \sqcap y)\sqcup(x \sqcap z)) \sqcup (z \sqcap y)$. 
For every $\{x,y\} \in M'$ we have that $t$ restricted to $\{x,y\}$ is the majority operation. For every other two-element subset $\{x,y\}$ of $D$ the operation $t$ restricted to $\{x,y\}$ is equal to $\pi_2$ or $\pi_3$.
Now let us define $m$ to be 
$$m(x,y,z)=\Mj_1(t(x,y,z),t(y,z,x),t(z,x,y)).$$
Since the operation $\Mj_1$ is idempotent, for every $\{x,y\} \in M'$ the operation $m$ restricted to $\{x,y\}$ is the majority operation. 
Moreover, if $\{x,y\}$ does not belong to $M'$ then $m$ restricted to $\{x,y\}$ is equal to $\Mj_1$ (with permuted arguments). But $\Mj_1$ restricted to $\{x,y\}$ is the majority operation. Therefore, $m$ is a majority operation on the whole domain~$D$. If there is a majority operation in the idempotent clone $\Pol^+(\G)$ then there is also an idempotent cyclic operation. This finishes the proof of the right-to-left implication.

The proof of the other implication consists of a sequence of claims and heavily relies on the results of~\cite{KZ}. 

Assume that $\Pol^{+}(\G)$ has an idempotent cyclic operation.
Let $M$ be a set of all two-element subsets $\{x,y\}$ of $D$ for which there exists no binary cost function $\rr \in \wRelClo(\G)$ such that
$$(x,y),(y,x) \in \Feas(\rr), \text{ and } \rr(x,x)+\rr(y,y)>\rr(x,y)+\rr(y,x).$$
We prove that $\G$ admits a conservative binary multimorphism $\langle \sqcap, \sqcup \rangle$ and a conservative ternary multimorphism $\langle \Mj_1,\Mj_2,\Mn_3\rangle$ such that:
\begin{itemize}
\item for every $\{x,y\} \in M$, $\langle \sqcap, \sqcup \rangle$ restricted to $\{x,y\}$ is an STP,
\item for every $\{x,y\} \not \in M$, $\langle \Mj_1,\Mj_2,\Mn_3\rangle$ restricted to $\{x,y\}$ is an MJN.
\end{itemize}

Consider the weighted algebra $(D, \wPol(\G))$. Observe that every two-element subset $\{x,y\} \subseteq D$ is a subuniverse of $D$ and let $\B$ be the subalgebra with universe $B = \{x,y\}$.

\begin{claim}\label{max-maj}
Every two-element weighted subalgebra $\B$ of $(D, \wPol(\G))$ contains the weighting $ \langle \min,\max \rangle$ or $\langle \Mjrty, \Mjrty, \Mnrty \rangle$ (we assume that $B=\{0,1\}$).
\end{claim}

\begin{proof}
The weighted algebra $(D, \wPol(\G))$ contains a weighting which assigns a positive weight to an idempotent cyclic operation. Therefore its weighted subalgebra $\B$ contains a weighting which assigns a positive weight to at least one operation that is not a projection, and hence it contains one of the nine weightings listed in Theorem~\ref{nonzero}. 
The first three of them are not idempotent.
Moreover, for each of the weightings: $ \langle \min, \min \rangle$, $ \langle \max, \max \rangle$, $ \langle \Mjrty, \Mjrty, \Mjrty \rangle$, $ \langle \Mnrty, \Mnrty, \Mnrty \rangle$ it is easy to find a $\{0,1\}$-valued unary cost function that is not improved by it. We conclude that $\B$ contains the weighting $ \langle \min,\max \rangle$ or $\langle \Mjrty, \Mjrty, \Mnrty \rangle$, which finishes the proof of Claim~\ref{max-maj}.
\end{proof}

\begin{claim}\label{soft}
Let $\{x,y\}$ be a two-element subset of $D$. There exists no binary cost function $\rr \in \wRelClo(\G)$ such that
$$(x,y),(y,x) \in \Feas(\rr), \text{ and } \rr(x,x)+\rr(y,y)>\rr(x,y)+\rr(y,x),$$ and at least one of the pairs $(x,x)$, $(y,y)$ belong to $\Feas(\rr)$.
\end{claim}

\begin{proof}
Consider the weighted subalgebra $\B$ of $(D, \wPol(\G)$ with the universe $B = \{x,y\}$. Assume that $\{x,y\}=\{0,1\}$. By Claim~\ref{max-maj} the weighted subalgebra $\B$ contains the weighting $ \langle \min,\max \rangle$ or $\langle \Mjrty, \Mjrty, \Mnrty \rangle$.

Suppose that $\B$ contains the weighting $\langle \Mjrty, \Mjrty, \Mnrty \rangle$ and let $\oo$ be the weighted polymorphims of $\G$ that induces $\langle \Mjrty, \Mjrty, \Mnrty \rangle$. Let $\rr$ be a binary cost function like in the statement of the claim. Without loss of generality assume that $(x,x) \in \Feas(\rr)$. Then
\begin{align*}
&\sum_{g \in \Pol_3(\G)} \oo(g) \cdot \rr(g((x,y),(y,x),(x,x))) = {1\over 3}\big( 2\rr(\Mjrty((x,y),(y,x),(x,x))) + \\ &+ \rr(\Mnrty((x,y),(y,x),(x,x))) -\rr(x,y) -\rr(y,x) - \rr(x,x)\big) = \\ 
&= {1\over 3}\big(2\rr(x,x) + \rr(y,y) -\rr(x,y) -\rr(y,x) - \rr(x,x)\big) = \\ &= {1\over 3}\big(\rr(x,x) + \rr(y,y) -\rr(x,y) -\rr(y,x)\big)>0.
\end{align*}
It follows that $\rr$ is not improved by $\oo$, so $\rr \not \in \wRelClo(\G)$.
Similarly we show that if $\B$ contains the weighting $\langle \min,\max \rangle$ and $\oo$ is the weighted polymorphims of $\G$ that induces $\langle \min,\max \rangle$, then $\rr$ is not improved by $\oo$.
\end{proof}

Claim~\ref{soft} together with Theorem 9 of~\cite{KZ} implies that $\G$ admits a conservative binary multimorphism $\langle \sqcap, \sqcup \rangle$ such that:
\begin{itemize}
\item for every $\{x,y\} \in M$, $\langle \sqcap, \sqcup \rangle$ restricted to $\{x,y\}$ is an STP,
\item if $\{x,y\} \not \in M$ then $\sqcap$ restricted to $\{x,y\}$ is equal to $\pi_1$ and $\sqcup$ restricted to $\{x,y\}$ is equal to $\pi_2$.
\end{itemize}

\begin{claim}\label{majority}
There exists an operation $m$ in $\Pol^+(\G)$ which is a majority operation.
\end{claim}

\begin{proof}
First we prove that there exists an operation $\Mj$ in $\Pol^+(\G)$ such that for every $\{x,y\} \not \in M$ the operation $\Mj$ restricted to $\{x,y\}$ is the majority operation.
To this end, take any $\{x,y\} \not \in M$. By the definition of $M$ there exists a binary cost function $\rr \in \wRelClo(\G)$ such that
$$(x,y),(y,x) \in \Feas(\rr), \text{ and } \rr(x,x)+\rr(y,y)>\rr(x,y)+\rr(y,x).$$ Moreover, by Claim~\ref{soft} none of the pairs $(x,x)$, $(y,y)$ belongs to $\Feas(\rr)$. Therefore:
\begin{enumerate}
\item\label{i1} there is no operation in $\Pol^+(\G)$ that restricted to $\{x,y\}$ is the $\max$ or $\min$ operation (such an operation would not even be a polymorphism of $\rr$), and
\item\label{i2} it follows from Claim~\ref{max-maj} that there exists an operation $f$ in $\Pol^+(\G)$ such that $f$ restricted to $\{x,y\}$ is the majority operation.
\end{enumerate}
By Proposition 3.1 of~\cite{Bulatov} if follows from the conditions~\ref{i1} and~\ref{i2} above that there exists an operation $\Mj$ in $\Pol^+(\G)$ such that for every $\{x,y\} \not \in M$ the operation $\Mj$ restricted to $\{x,y\}$ is the majority operation.

Let $t(x,y,z) = ((x \sqcap y)\sqcup(x \sqcap z)) \sqcup (z \sqcap y)$. 
For every $\{x,y\} \in M$ we have that $t$ restricted to $\{x,y\}$ is the majority operation. For every other two-element subset $\{x,y\}$ of $D$ the operation $t$ restricted to $\{x,y\}$ is equal to $\pi_3$.
Now let us define $m$ to be 
$$m(x,y,z)=\Mj(t(x,y,z),t(y,z,x),t(z,x,y)).$$
Since $\Mj$ is idempotent, for every $\{x,y\} \in M$ the operation $m$ restricted to $\{x,y\}$ is the majority operation. 
Moreover, if $\{x,y\}$ does not belong to $M$ then $m$ restricted to $\{x,y\}$ is equal to $\Mj$ (with permuted arguments). But $\Mj$ restricted to $\{x,y\}$ is the majority operation. Therefore, $m$ is a majority operation on the whole domain~$D$.
\end{proof}

By~\cite{KZ} if follows from Claims~\ref{soft} and~\ref{majority} that $\G$ admits a conservative ternary multimorphism $\langle \Mj_1,\Mj_2,\Mn_3\rangle$ such that for every $\{x,y\} \not \in M$, $\langle \Mj_1,\Mj_2,\Mn_3\rangle$ restricted to $\{x,y\}$ is an MJN, which finishes the proof of Proposition~\ref{charcon}.
\end{proof}

\subsection{Infinite Constraint Languages.}

We finish this section by showing that a variation of Theorem~\ref{taylor} holds for infinite constraint languages. 
For that, we need to introduce \emph{fractional polymorphisms} which correspond to weighted polymorphisms that can take real values.

An $m$-ary \emph{fractional operation} $\omega$ on $D$ is a probability distribution on $\Oo_D^{(m)}$. 
As for weightings, the set of operations to which a fractional operation $\omega$ assigns a positive probability is called the \emph{support} of $\omega$ and denoted $\supp(\omega)$.

\begin{definition} 
An $m$-ary fractional operation $\omega$ on $D$ is a \emph{fractional polymorphism} of a cost function $\rr$ if, for any list of $r$-tuples $\bf x_1, \dots, x_m \in \Feas(\rr)$, we have
$$\sum_{g \in \supp(\omega)} \omega(g) \rr(g({\bf x_1},\dots,{\bf x_m})) \leq {1\over m} (\rr({\bf x_1})+\dots+\rr({\bf x_m})).$$
\end{definition}

For a constraint language $\G$ we denote by $\fPol(\Gamma)$ the set of those fractional operations that are fractional polymorphisms of all cost functions $\rr \in \Gamma$.
Let $\fPol^+(\G) = \{g \in \Oo_D \ | \ g \in \supp(\omega), \ \omega \in \fPol(\Gamma)\}$. 
It is easy to see that $\fPol^+(\G)$ is a clone (the proof is similar to that of Proposition~\ref{prop:posclo}).

\begin{proposition}\label{prop:positive}
Let $\Gamma$ be a finite constraint language. Then $\fPol^+(\G) = \Pol^+(\G)$.
\end{proposition}

\begin{proof}
Obviously $\Pol^+(\G) \subseteq \fPol^+(\G)$. Let $f \in \fPol^+(\G)$. 
This can be equivalently expressed by saying that some LP with rational coefficients has a solution. 
It is well known that every LP with rational coefficients has an optimal solution with rational coefficients (see e.g.~\cite{Schrijver:1986:TLI:17634}). 
It follows that $f \in \Pol^+(\G)$.
\end{proof}

\begin{proposition}
There exists an infinite valued constraint language $\Gamma$ such that $\Pol^+(\G) \varsubsetneq \fPol^+(\G)$. 
\end{proposition}

An example of such language is given in~\cite{TZ}.

To deal with infinite languages we slightly modify the notion of a core. 
We say that a valued constraint language $\G$ is a \emph{core} if all operations in $\fPol^+_1(\G)$ are bijective. 
It follows from Proposition~\ref{prop:positive} that for finite languages this definition coincides with the old one. 
Moreover, the following proposition states that all positively weighted unary polymorphisms can be captured in a finite part of $\G$.

\begin{proposition}\label{prop:fcore}
If a valued constraint language $\G$ is a core then there exists a finite language $\G' \subseteq \G$ such that $\fPol_1^+(\G) = \fPol_1^+(\G')$.
\end{proposition}

We will use the following lemma (Lemma 7 from~\cite{TZ}), which is an immediate consequence of the fact that the space of $m$-ary fractional operations on a fixed domain $D$ is compact.

\begin{lemma}\label{lem:compact}
Let $\G$ be a valued constraint language and let $O \subseteq \Oo_D^{(m)}$. 
If for every finite $\G' \subseteq \G$ there exists a fractional polymorphism with support in $O$, then $\G$ has a fractional polymorphism with support in $O$.
\end{lemma}

\begin{proof}
(of Proposition~\ref{prop:fcore})
Let $\G$ be a core valued constraint language and let $B$ be the set of such bijective operations $f$ on $D$ that, for all cost functions $\rr \in \G$, satisfy $\rr \circ f = \rr$. 
A proof analogous to the one of Proposition~\ref{unary} shows that a unary fractional operation $\oo$ is a fractional polymorphism of $\G$ if and only if $\supp(\oo) \subseteq B$. 

Take a language $\G'$ such that $\fPol_1^+(\G') \neq \fPol_1^+(\G)$. 
This means that there exists a fractional polymorphism $\oo$ of $\G'$ such that there is some operation $g$ from outside $B$ in $\supp(\oo)$. 
Then there exists a fractional polymorphism $\oo'$ of $\G'$ such that $\supp(\oo') \subseteq \Oo_D^{(1)} \setminus B$. 
We construct it from the fractional polymorphism $\oo$ by assigning probability $0$ to all operations in $B$ and scaling to get a proper probability distribution. 
It is not difficult to check that the obtained fractional operation $\oo'$ is a fractional polymorphism of $\G'$.

Now suppose that none of the finite $\G' \subseteq \G$ satisfies $\fPol_1^+(\G) = \fPol_1^+(\G')$. 
Then every finite $\G' \subseteq \G$ has a fractional polymorphism with support in $\Oo_D^{(1)} \setminus B$. 
It follows from Lemma~\ref{lem:compact} that $\G$ has a fractional polymorphism with support in $\Oo_D^{(1)} \setminus B$, which is a contradiction.
\end{proof}

A proof analogous to the one of Proposition~\ref{core} shows that every constraint language $\G$ has a computationally equivalent language $\G'$ which is a core (in the new sense).

\begin{proposition}\label{corer}
For every valued constraint language $\G$ there exists a core language $\G'$, 
such that the valued constraint language $\G$ is tractable if and only if $\G'$ is tractable, 
and it is NP-hard if and only if $ \G'$ is NP-hard.
\end{proposition}

Moreover, using Proposition~\ref{prop:fcore}, we can show that constants can be added to a core language~(the proposition implies that, for a core language $\G$, the relation characterizing $\fPol_1^+(\G)$ belongs to $\wRelClo(\G)$). 
The proof follows along the same lines as that of Proposition~\ref{rigid core}.

\begin{proposition}
Let $\G$ be a valued constraint language which is a core.
The valued constraint language $\G_c$ is a rigid core.
Moreover, $\G$ is tractable if and only if $\G_c$ is tractable,
and $\G$ is NP-hard if and only if $\G_c$ is NP-hard.
\end{proposition}

\begin{theorem}
Let $\G$ be a core valued constraint language. 
If, for every finite $\G'\subset\G$, the set $\fPol^{+}(\G')$ has a Taylor operation, then so does $\fPol^+(\G)$.
\end{theorem}

\begin{proof}
Consider a finite $\G' \subseteq \G$. 
Since $\G$ is a core, by Proposition~\ref{prop:fcore} there exists a finite $\hat{\G} \subseteq \G$ which is a core. 
By $\overline{\G'}$ we denote the language $\G' \cup \hat{\G}$. 
It is a finite subset of $\G$, and hence it has a Taylor operation. 
Fix some prime number $p > |D|$. 
It follows from Theorem~\ref{taylor} and Theorem~\ref{thm:BK} that $\Pol^{+}(\overline{\G'})$ has an idempotent cyclic operation of arity $p$, 
and therefore by Proposition~\ref{char} $\overline{\G'}$ admits an idempotent cyclic weighted polymorphism $\oo'$ of arity $p$. 

Let $O$ denote the set of $p$-ary idempotent cyclic operations on $D$. 
From the weighted polymorphism $\oo'$ it is easy to construct a fractional polymorphism $\overline{\oo'}$ such that $\supp({\overline{\oo'}}) \subseteq O$. 
Since $\G' \subseteq \overline{\G'}$, we have that $\overline{\oo'}$ is a fractional polymorphism of $\G'$ with support in $O$. 
This holds for every finite $\G' \subseteq \G$. 
Hence, by Lemma~\ref{lem:compact} there exists a fractional polymorphism $\oo$ of $\G$ with support in $O$. 
\end{proof}

It immediately follows that if $\G$ does not have a Taylor operation in $\fPol^+(\G)$ then it is NP-hard.

\begin{corollary}
Let $\G$ be a core valued constraint language. If $\fPol^{+}(\G)$ does not have a Taylor operation, then $\G$ is NP-hard.
\end{corollary}

\bibliographystyle{plain}
\bibliography{bib}

%
%

\end{document}